\documentclass[copyright,creativecommons]{eptcs}
\providecommand{\event}{EXPRESS/SOS 2025}
\hypersetup{hidelinks}
\usepackage{graphicx}
\usepackage{iftex}
\usepackage{latexsym,amsmath,amssymb,color}

\ifpdf
  \usepackage{underscore}         % Only needed if you use pdflatex.
  \usepackage[T1]{fontenc}        % Recommended with pdflatex
\else
  \usepackage{breakurl}           % Not needed if you use pdflatex only.
\fi
\usepackage[shortlabels]{enumitem}
\setlist{nolistsep}  % no useless space in lists
\DeclareSymbolFont{frenchscript}{OMS}{ztmcm}{m}{n}
\DeclareMathSymbol{\Pow}{\mathord}{frenchscript}{80}  % powerset
\DeclareMathSymbol{\B}{\mathrel}{frenchscript}{66}    % bisimulation
\DeclareMathAlphabet{\mathcal}{OMS}{cmsy}{m}{n}       % latex mathcal default
\DeclareMathAlphabet{\mathbbm}{U}{bbm}{m}{n}          % blackboard bold
\newcommand{\IT}{\mathbbm{T}}                         % open terms
\newcommand{\IN}{\mathbbm{N}}                         % natural numbers
\newcommand{\RS}{{\cal S}}                            % a recursive specification
\newcommand{\T}{{\rm T}}                              % closed terms
\newcommand{\TS}{\mathcal{T}}                         % TSS
\newcommand{\R}{\mathcal{R}}                          % set of rules in a TSS
\makeatletter
\newif\if@qeded
\global\@qededtrue
\def\qed{\hfill$\Box$\global\@qededtrue}
\def\qedneeded{\global\@qededfalse}
\def\qedifneeded{\if@qeded\else\qed\fi}
\makeatother
\newtheorem{defi}{Definition}
\newtheorem{theo}{Theorem}
\newtheorem{lemm}{Lemma}
\newtheorem{prop}{Proposition}
\newtheorem{coro}{Corollary}
\newtheorem{exam}{Example}
\newtheorem{obse}{Observation}

\newenvironment{definition}[1]{\begin{defi} \rm \label{df:#1} }{\end{defi}}
\newenvironment{definitionA}[2]{\begin{defi}[#1] \rm \label{df:#2} }{\end{defi}}
\newenvironment{theorem}[1]{\begin{theo} \rm \label{thm:#1} }{\end{theo}}
\newenvironment{theoremA}[2]{\begin{theo}[#1] \rm \label{thm:#2} }{\end{theo}}
\newenvironment{lemma}[1]{\begin{lemm} \rm \label{lem:#1} }{\end{lemm}}
\newenvironment{lemmaA}[2]{\begin{lemm}[#1] \rm \label{lem:#2} }{\end{lemm}}
\newenvironment{proposition}[1]{\begin{prop} \rm \label{pr:#1} }{\end{prop}}

\newenvironment{corollary}[1]{\begin{coro} \rm \label{cor:#1} }{\end{coro}}
\newenvironment{example}[1]{\begin{exam} \rm \label{ex:#1} }{\end{exam}}

\newenvironment{observation}[1]{\begin{obse} \rm \label{obs:#1} }{\end{obse}}
\newenvironment{proof}{\qedneeded\begin{trivlist} \item[\hspace{\labelsep}\bf Proof:]}
                      {\qedifneeded\end{trivlist}}
\newcommand{\Sec}[1]{Section~\ref{sec:#1}}
\newcommand{\df}[1]{Definition~\ref{df:#1}}
\newcommand{\thm}[1]{Theorem~\ref{thm:#1}}
\newcommand{\lem}[1]{Lemma~\ref{lem:#1}}
\newcommand{\pr}[1]{Proposition~\ref{pr:#1}}
\newcommand{\ex}[1]{Example~\ref{ex:#1}}
\newcommand{\obs}[1]{Observation~\ref{obs:#1}}
\makeatletter
\def\comesfrom{\@transition\leftarrowfill}
\def\goesto{\@transition\rightarrowfill}
\def\ngoesto{\@transition\nrightarrowfill}
\def\Goesto{\@transition\Rightarrowfill}
\def\nGoesto{\@transition\nRightarrowfill}
\def\xmapsto{\@transition\mapstofill}
\def\nxmapsto{\@transition\nmapstofill}
\def\@transition#1{\@@transition{#1}}
\newbox\@transbox
\newbox\@arrowbox
\newbox\@downbox
\def\@@transition#1#2%
   {\setbox\@transbox\hbox
      {\vrule height 1.5ex depth .8ex width 0ex\hskip0.25em$\scriptstyle#2$\hskip0.25em}
   \ifdim\wd\@transbox<1.5em
      \setbox\@transbox\hbox to 1.5em{\hfil\box\@transbox\hfil}\fi
   \setbox\@arrowbox\hbox to \wd\@transbox{#1}
   \ht\@arrowbox\z@\dp\@arrowbox\z@
   \setbox\@transbox\hbox{$\mathop{\box\@arrowbox}\limits^{\box\@transbox}$}
   \dp\@transbox\z@\ht\@transbox 10pt
   \mathrel{\box\@transbox}}
\def\nrightarrowfill{$\m@th\mathord-\mkern-6mu%
  \cleaders\hbox{$\mkern-2mu\mathord-\mkern-2mu$}\hfill
  \mkern-6mu\mathord\not\mkern-2mu\mathord\rightarrow$}
\def\Rightarrowfill{$\m@th\mathord=\mkern-6mu%
  \cleaders\hbox{$\mkern-2mu\mathord=\mkern-2mu$}\hfill
  \mkern-6mu\mathord\Rightarrow$}
\def\nRightarrowfill{$\m@th\mathord=\mkern-6mu%
  \cleaders\hbox{$\mkern-2mu\mathord=\mkern-2mu$}\hfill
  \mkern-6mu\mathord\not\mathord\Rightarrow$}
\def\mapstofill{$\m@th\mathord\mapstochar\mathord-\mkern-6mu%
  \cleaders\hbox{$\mkern-2mu\mathord-\mkern-2mu$}\hfill
  \mkern-6mu\mathord\rightarrow$}
\def\nmapstofill{$\m@th\mathord\mapstochar\mathord-\mkern-6mu%
  \cleaders\hbox{$\mkern-2mu\mathord-\mkern-2mu$}\hfill
  \mkern-6mu\mathord\not\mkern-2mu\mathord\rightarrow$}
\makeatother %%%%% end of arrows definition
\newcommand{\goesnotto}[1]{{\ngoesto{#1\;}}}          % negated arrow
\newcommand{\goto}[1]{\stackrel{#1}{\longrightarrow}} % transition
\newcommand{\hoto}[1]{\mathbin{\stackrel{#1}{\raisebox{0pt} % hypotetical
        [3pt][0pt]{$\scriptstyle--\rightarrow$}}}}          % transition
\newcommand{\gonotto}[1]{\mbox{$\,\,\,\,\not\!\!\!\!\stackrel{#1~}{\longrightarrow}$}} % negated tr.
\newcommand{\honotto}[1]{\mbox{$\,\,\,\,\,\,\,\not\!\!\!\!\!\!\stackrel{#1~}{\raisebox{0pt}
        [3pt][0pt]{$\scriptstyle--\rightarrow$}}$}}   % negated hypothetial transition                              
\newcommand{\plat}[1]{\raisebox{0pt}[0pt][0pt]{#1}}   % no vertical space
\newcommand{\rec}[1]{\plat{$			      % recursion
	\stackrel{\mbox{\tiny $/$}}
	{\raisebox{-.3ex}[.3ex]{\tiny $\backslash$}}
	\!\!#1\!\!
	\stackrel{\mbox{\tiny $\backslash$}}
	{\raisebox{-.3ex}[.3ex]{\tiny $/$}} $}}
\newcommand{\leftm}{\mathbin{\lfloor\hspace{-3pt}\lfloor}} % left merge
\newcommand{\Var}{{\it Var}}                          % variables
\newcommand{\var}{{\it var}}                          % variables occurring in a term
\newcommand{\ar}[1]{\mathit{ar}(#1)}                  % arity of operator
\newcommand{\E}{E}                                    % expression
\newcommand{\al}{a}                                   % action
\def\titlerunning{Unique Solutions of Guarded Recursive Equations}
\def\authorrunning{R.J. van Glabbeek}
\title\titlerunning
\author{Rob van Glabbeek%
  \thanks{Supported by Royal Society Wolfson Fellowship RSWF\textbackslash R1\textbackslash 221008}
  \,\href{https://orcid.org/0000-0003-4712-7423}{\includegraphics[scale=.04]{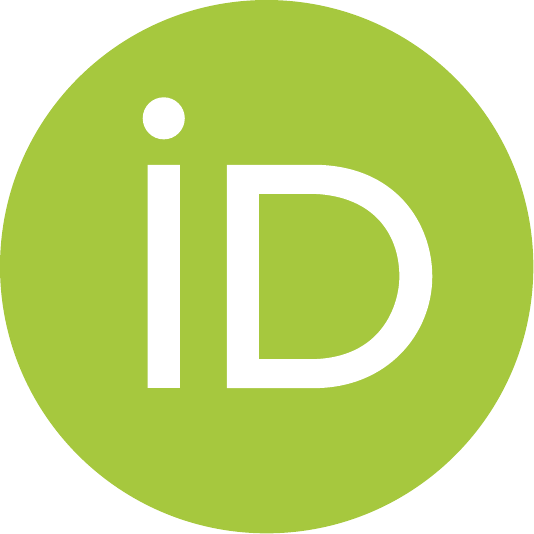}}
  \institute{School of Informatics, University of Edinburgh, UK}
  \institute{School of Computer Science and Engineering,
    University of New South Wales, Sydney, Australia}
  \email{rvg@cs.stanford.edu}
}
\begin{document}
\maketitle

\begin{abstract}
This paper shows that guarded systems of recursive equations have unique solutions up to strong
bisimilarity for any process algebra with a structural operation semantics in the ready simulation
format. A similar result holds for simulation equivalence, for ready simulation equivalence and for the
(ready) simulation preorder. As a consequence, these equivalences and preorders are full (pre)
congruences for guarded recursion. Moreover, the unique-solutions result yields a sound and
ground-complete axiomatisation of strong bisimilarity for any finitary GSOS language.

\end{abstract}

\section{Introduction}

In many process algebras, including CCS \cite{Mi90ccs}, CSP \cite{BHR84} and ACP \cite{BW90},
infinite processes are specified by systems $\RS$ of recursive equations $\{X = \RS_X \mid X \in V_\RS\}$.
Such systems are called \emph{recursive specifications}. Here $V_\RS$ is a set of recursion
variables specific to a recursive specification, and each $X \in V_\RS$ is defined to behave as
the expression $\RS_X$. Crucially, $X$ itself, and other variables from $V_\RS$, may occur in the
expression $\RS_X$. An important tool in equational reasoning is that a large class of recursive
specifications have unique solutions. Here a \emph{solution} is a tuple $\vec{P}$ of processes, one
for each variable $X$ in the set $V_\RS$, such that filling in the $X$-component of $\vec{P}$ for
$X$, at both the left- and the right-hand side of the recursive equations in $\RS$, yields a set of
semantically valid statements.

\begin{example}{recursive specification}
Using the syntax of CCS \cite{Mi90ccs}, let $\RS$ be the recursive specification consisting of the
two recursive equations $\{X = a.X +b.Y ~,~ Y = c.X + d.Y\}$. In the model of labelled transition
systems \cite{Mi90ccs}, the pair $(P,Q)$ of states in the labelled transition system depicted below
constitutes a solution of $\RS$.\vspace{-5pt}

\expandafter\ifx\csname graph\endcsname\relax
   \csname newbox\expandafter\endcsname\csname graph\endcsname
\fi
\ifx\graphtemp\undefined
  \csname newdimen\endcsname\graphtemp
\fi
\expandafter\setbox\csname graph\endcsname
 =\vtop{\vskip 0pt\hbox{%
\pdfliteral{
q [] 0 d 1 J 1 j
0.576 w
0.576 w
39.6 -14.4 m
39.6 -16.388225 37.988225 -18 36 -18 c
34.011775 -18 32.4 -16.388225 32.4 -14.4 c
32.4 -12.411775 34.011775 -10.8 36 -10.8 c
37.988225 -10.8 39.6 -12.411775 39.6 -14.4 c
S
Q
}%
    \graphtemp=.5ex
    \advance\graphtemp by 0.380in
    \rlap{\kern 0.500in\lower\graphtemp\hbox to 0pt{\hss $P$\hss}}%
\pdfliteral{
q [] 0 d 1 J 1 j
0.576 w
0.072 w
q 0 g
27.432 -7.56 m
33.48 -11.88 l
26.064 -10.872 l
27.432 -7.56 l
B Q
0.576 w
33.48 -16.92 m
27.54 -19.26 l
23.5008 -20.8512 19.296 -21.6 14.4 -21.6 c
9.504 -21.6 7.2 -19.296 7.2 -14.4 c
7.2 -9.504 9.504 -7.2 14.4 -7.2 c
19.296 -7.2 23.3856 -7.90272 27.18 -9.396 c
32.76 -11.592 l
S
Q
}%
    \graphtemp=.5ex
    \advance\graphtemp by 0.200in
    \rlap{\kern 0.000in\lower\graphtemp\hbox to 0pt{\hss $a$\hss}}%
\pdfliteral{
q [] 0 d 1 J 1 j
0.576 w
111.6 -14.4 m
111.6 -16.388225 109.988225 -18 108 -18 c
106.011775 -18 104.4 -16.388225 104.4 -14.4 c
104.4 -12.411775 106.011775 -10.8 108 -10.8 c
109.988225 -10.8 111.6 -12.411775 111.6 -14.4 c
S
Q
}%
    \graphtemp=.5ex
    \advance\graphtemp by 0.380in
    \rlap{\kern 1.500in\lower\graphtemp\hbox to 0pt{\hss $Q$\hss}}%
\pdfliteral{
q [] 0 d 1 J 1 j
0.576 w
0.072 w
q 0 g
98.064 -16.704 m
105.48 -16.92 l
98.784 -20.232 l
98.064 -16.704 l
B Q
0.576 w
98.412476 -18.508989 m
78.57244 -22.210851 58.175807 -21.700826 38.545583 -17.011994 c
S
Q
}%
    \graphtemp=.5ex
    \advance\graphtemp by 0.400in
    \rlap{\kern 1.000in\lower\graphtemp\hbox to 0pt{\hss $b$\hss}}%
\pdfliteral{
q [] 0 d 1 J 1 j
0.576 w
0.072 w
q 0 g
45.936 -12.096 m
38.52 -11.88 l
45.216 -8.568 l
45.936 -12.096 l
B Q
0.576 w
45.587475 -10.36302 m
65.42751 -6.661151 85.824143 -7.17117 105.454368 -11.859994 c
S
Q
}%
    \graphtemp=.5ex
    \advance\graphtemp by 0.000in
    \rlap{\kern 1.000in\lower\graphtemp\hbox to 0pt{\hss $c$\hss}}%
\pdfliteral{
q [] 0 d 1 J 1 j
0.576 w
0.072 w
q 0 g
117.936 -10.872 m
110.52 -11.88 l
116.568 -7.56 l
117.936 -10.872 l
B Q
0.576 w
110.52 -16.92 m
116.46 -19.26 l
120.4992 -20.8512 124.704 -21.6 129.6 -21.6 c
134.496 -21.6 136.8 -19.296 136.8 -14.4 c
136.8 -9.504 134.496 -7.2 129.6 -7.2 c
124.704 -7.2 120.6144 -7.90272 116.82 -9.396 c
111.24 -11.592 l
S
Q
}%
    \graphtemp=.5ex
    \advance\graphtemp by 0.200in
    \rlap{\kern 2.000in\lower\graphtemp\hbox to 0pt{\hss $d$\hss}}%
    \hbox{\vrule depth0.400in width0pt height 0pt}%
    \kern 2.000in
  }%
}%

\centerline{\box\graph}
\vspace{6pt}

\noindent
In fact, this is the only solution of $\RS$ up to strong bisimilarity \cite{Mi90ccs}.
\end{example}
Unfortunately, not all recursive specifications have unique solutions.
The recursive equation $X = X$, for instance, allows any process as a solution.
\textsc{Milner} \cite{Mi90ccs} formulated a syntactic criterion on recursive specifications called
\emph{guardedness}, and showed that in the process algebra CCS, systems of guarded recursive equations have
unique solutions up to strong bisimilarity. In the setting of ACP \cite{BBK87a,BW90} two reasoning
principles are defined: the \emph{Recursive Definition Principle} (RDP) states that any recursive
specification has a solution, and the \emph{Recursive Specification Principle} (RSP) states that
guarded recursive specifications have at most one solution, up to a semantic equivalence.
Whether these principles hold may depend on the process algebra and on the semantic equivalence employed.
First of all it takes a semantic equivalence to determine whether a given tuple $\vec{P}$ is a
solution of a recursive specification $\RS$. Secondly, if $\RS$ is guarded and $\vec{P}$ and
$\vec{Q}$ both are solutions, RSP demands that $\vec{P}$ and $\vec{Q}$ are semantically equivalent.
RDP and RSP are equational and conditional equational laws that can be added as axioms to the
equational theories of various process algebras. Doing so has been a crucial step in achieving
complete axiomatisations \cite{Mi84,BK88,Fok00,GM20}---see also \Sec{complete}.

Subtly different definitions of guardedness appear in the literature. In \cite{Mi90ccs} it was
required that within the bodies $\RS_X$ of recursive equations the recursion variables $Y \in V_\RS$
may appear only within the scope of a CCS action prefixing operator $a.\_$. Here this action
\emph{guards} the recursion variable. In \ex{recursive specification}, $\RS$ is guarded in this
sense. In other process algebras one might allow other operators as
guards. In ACP \cite{BW90} for instance, the left merge $\leftm$ is a good candidate. However, it
guards only its second argument: the recursive specification $X = b \leftm X$ is guarded, but
$Y = Y\leftm b$ is not.
In general, a suitable definition of guardedness is whatever syntactic criterion allows one to
obtain RSP for a chosen semantic equivalence.
Proofs of RSP \cite{Mi90ccs,BBK87a} tend to depend to a large extent on the syntax and semantics of
the process algebra involved.

The goal of this paper is to give a proof of RSP for strong bisimulation semantics for a large class
of process algebras, including CCS and ACP, at once. I do this by characterising a format for
transition system specifications in structural operational semantics, proposing the most general
definition of guardedness I can find in that context, and formally show that all process algebras
whose operational semantics fits that format satisfy RSP, for strong bisimilarity. They also satisfy
RDP, but that is fairly trivial.

The most general format for standard process algebras (without probabilities, time, name-binding or
other special features) that guarantees strong bisimilarity to be congruence is the
ntyft/ntyxt format of \textsc{Groote} \cite{Gr93}, augmented with recursion as in \cite{vG17b}.
Let $f$ (pronounced ``after $a$'') be a unary operator defined through the transition rules\vspace{-2ex}
\[\frac{x \goto{a} y ~~~ y \goto{b} z}{f(x) \goto{b} z}\]
for the specific action $a$ and each action $b$. It has lookahead, in the sense that the initial
actions of $f$ are determined by looking ahead at the second action performed by its argument.
Now the recursive equation $X = f(a.X)$ has any process as solution, even though the recursion
variable $X$ appears guarded (by the guard $a$) in the body $f(a.X)$ of the recursion equation.
Because of this I rule out lookahead in this paper. Aiming for maximal generality, I will obtain RSP
for strong bisimilarity in the ready simulation format \cite{vG93d}, which is the ntyft/ntyxt format
without lookahead, augmented with recursion.

Most applications of this work will fall in the GSOS format, which restricts the ready simulation
format mainly by allowing only variables to occur in premises, except in the rules for recursion.
In this context I favour a notion of guardedness due to \textsc{Vaandrager} \cite{Va93}: an $n$-ary
operator $f$ guards its $i^{\rm th}$ argument iff there is no transition rule of the form
  \[\frac{H}{f(x_1,\dots,x_n) \goto\al E}\]
such that the variable $x_i$ occurs as the left-hand side of one of the premises in $H$.
I will generalise this idea to the ready simulation format by allowing $x_i$ to occur in the
left-hand sides of premises (which may be complex terms) but only in guarded positions.

After establishing RSP for strong bisimilarity, I observe that by means of an almost identical proof
it can also be shown for other strong semantic equivalences at the branching time side of the linear
time -- branching time spectrum \cite{vG01}, namely simulation equivalence (in the absence of
negative premises) and ready simulation equivalence. Naturally I would like to show it for preorders as well.
Using the notation $\vec{P} \equiv \RS[\vec{P}]$ to indicate that the tuple $\vec{P}$ is a solution
of the recursive specification $\RS$ up to the semantic equivalence $\equiv$, there are two ways to
replace $\equiv$ by a preorder $\sqsubseteq$. When $\vec{P} \sqsubseteq \RS[\vec{P}]$ holds I call
$\vec{P}$ a presolution of $\RS[\vec{P}]$, and when $\RS[\vec{P}] \sqsubseteq \vec{P}$ holds I call
it a postsolution. The form of RSP for $\equiv$ simply says that if $\vec{P}$ and $\vec{Q}$ are both
solutions of $\RS$ then $\vec{P} \equiv \vec{Q}$. It turns out that no interesting conclusions can be
drawn from the assumption that  $\vec{P}$ and $\vec{Q}$ are both presolutions, or both postsolutions.
However, an asymmetric form of RSP holds if $\vec{P}$ is a presolution and $\vec{Q}$ a postsolution.
In that case $\vec{P} \sqsubseteq \vec{Q}$, at least when taking for $\sqsubseteq$ the strong
simulation preorder or the ready simulation preorder. I thus formulate RSP as
\begin{equation}
\frac{\vec{P} \equiv \RS[\vec{P}] \qquad \RS[\vec{Q}] \equiv \vec{Q}}{\vec{P} \equiv \vec{Q}}
\qquad
 \frac{\vec{P} \sqsubseteq \RS[\vec{P}] \qquad \RS[\vec{Q}] \sqsubseteq \vec{Q}}{\vec{P} \sqsubseteq
   \vec{Q}}.
\label{RSP}
\end{equation}
\hypertarget{kernel}{In most cases a semantic equivalence $\equiv$ (say simulation equivalence)
arises as the kernel of a preorder, meaning that $P \equiv Q$ iff $P \sqsubseteq Q$ and $Q \sqsubseteq P$.}
In this setting, RSP for $\equiv$ is a simple consequence of RSP for $\sqsubseteq$.
Thus, one merely needs to prove the latter.

For linear time equivalences, a generic counterexample to RSP is given by these processes:\\
\expandafter\ifx\csname graph\endcsname\relax
   \csname newbox\expandafter\endcsname\csname graph\endcsname
\fi
\ifx\graphtemp\undefined
  \csname newdimen\endcsname\graphtemp
\fi
\expandafter\setbox\csname graph\endcsname
 =\vtop{\vskip 0pt\hbox{%
    \graphtemp=.5ex
    \advance\graphtemp by 2.224in
    \rlap{\kern 2.067in\lower\graphtemp\hbox to 0pt{\hss $P'$\hss}}%
\pdfliteral{
q [] 0 d 1 J 1 j
0.576 w
0.576 w
150.264 -4.248 m
150.264 -5.04329 149.61929 -5.688 148.824 -5.688 c
148.02871 -5.688 147.384 -5.04329 147.384 -4.248 c
147.384 -3.45271 148.02871 -2.808 148.824 -2.808 c
149.61929 -2.808 150.264 -3.45271 150.264 -4.248 c
S
121.896 -21.24 m
121.896 -22.03529 121.25129 -22.68 120.456 -22.68 c
119.66071 -22.68 119.016 -22.03529 119.016 -21.24 c
119.016 -20.44471 119.66071 -19.8 120.456 -19.8 c
121.25129 -19.8 121.896 -20.44471 121.896 -21.24 c
S
121.896 -43.92 m
121.896 -44.71529 121.25129 -45.36 120.456 -45.36 c
119.66071 -45.36 119.016 -44.71529 119.016 -43.92 c
119.016 -43.12471 119.66071 -42.48 120.456 -42.48 c
121.25129 -42.48 121.896 -43.12471 121.896 -43.92 c
S
150.264 -4.248 m
150.264 -5.04329 149.61929 -5.688 148.824 -5.688 c
148.02871 -5.688 147.384 -5.04329 147.384 -4.248 c
147.384 -3.45271 148.02871 -2.808 148.824 -2.808 c
149.61929 -2.808 150.264 -3.45271 150.264 -4.248 c
S
150.264 -32.616 m
150.264 -33.41129 149.61929 -34.056 148.824 -34.056 c
148.02871 -34.056 147.384 -33.41129 147.384 -32.616 c
147.384 -31.82071 148.02871 -31.176 148.824 -31.176 c
149.61929 -31.176 150.264 -31.82071 150.264 -32.616 c
S
121.896 -60.912 m
121.896 -61.70729 121.25129 -62.352 120.456 -62.352 c
119.66071 -62.352 119.016 -61.70729 119.016 -60.912 c
119.016 -60.11671 119.66071 -59.472 120.456 -59.472 c
121.25129 -59.472 121.896 -60.11671 121.896 -60.912 c
S
150.264 -60.912 m
150.264 -61.70729 149.61929 -62.352 148.824 -62.352 c
148.02871 -62.352 147.384 -61.70729 147.384 -60.912 c
147.384 -60.11671 148.02871 -59.472 148.824 -59.472 c
149.61929 -59.472 150.264 -60.11671 150.264 -60.912 c
S
178.632 -60.912 m
178.632 -61.70729 177.98729 -62.352 177.192 -62.352 c
176.39671 -62.352 175.752 -61.70729 175.752 -60.912 c
175.752 -60.11671 176.39671 -59.472 177.192 -59.472 c
177.98729 -59.472 178.632 -60.11671 178.632 -60.912 c
S
121.896 -89.28 m
121.896 -90.07529 121.25129 -90.72 120.456 -90.72 c
119.66071 -90.72 119.016 -90.07529 119.016 -89.28 c
119.016 -88.48471 119.66071 -87.84 120.456 -87.84 c
121.25129 -87.84 121.896 -88.48471 121.896 -89.28 c
S
150.264 -89.28 m
150.264 -90.07529 149.61929 -90.72 148.824 -90.72 c
148.02871 -90.72 147.384 -90.07529 147.384 -89.28 c
147.384 -88.48471 148.02871 -87.84 148.824 -87.84 c
149.61929 -87.84 150.264 -88.48471 150.264 -89.28 c
S
178.632 -89.28 m
178.632 -90.07529 177.98729 -90.72 177.192 -90.72 c
176.39671 -90.72 175.752 -90.07529 175.752 -89.28 c
175.752 -88.48471 176.39671 -87.84 177.192 -87.84 c
177.98729 -87.84 178.632 -88.48471 178.632 -89.28 c
S
150.264 -117.648 m
150.264 -118.44329 149.61929 -119.088 148.824 -119.088 c
148.02871 -119.088 147.384 -118.44329 147.384 -117.648 c
147.384 -116.85271 148.02871 -116.208 148.824 -116.208 c
149.61929 -116.208 150.264 -116.85271 150.264 -117.648 c
S
178.632 -117.648 m
178.632 -118.44329 177.98729 -119.088 177.192 -119.088 c
176.39671 -119.088 175.752 -118.44329 175.752 -117.648 c
175.752 -116.85271 176.39671 -116.208 177.192 -116.208 c
177.98729 -116.208 178.632 -116.85271 178.632 -117.648 c
S
178.632 -146.016 m
178.632 -146.81129 177.98729 -147.456 177.192 -147.456 c
176.39671 -147.456 175.752 -146.81129 175.752 -146.016 c
175.752 -145.22071 176.39671 -144.576 177.192 -144.576 c
177.98729 -144.576 178.632 -145.22071 178.632 -146.016 c
S
Q
}%
    \graphtemp=.5ex
    \advance\graphtemp by 0.256in
    \rlap{\kern 2.008in\lower\graphtemp\hbox to 0pt{\hss $a$\hss}}%
    \graphtemp=.5ex
    \advance\graphtemp by 0.098in
    \rlap{\kern 1.870in\lower\graphtemp\hbox to 0pt{\hss $a$\hss}}%
    \graphtemp=.5ex
    \advance\graphtemp by 0.634in
    \rlap{\kern 1.827in\lower\graphtemp\hbox to 0pt{\hss $a$\hss}}%
    \graphtemp=.5ex
    \advance\graphtemp by 0.650in
    \rlap{\kern 2.335in\lower\graphtemp\hbox to 0pt{\hss $a$\hss}}%
    \graphtemp=.5ex
    \advance\graphtemp by 0.650in
    \rlap{\kern 2.008in\lower\graphtemp\hbox to 0pt{\hss $a$\hss}}%
    \graphtemp=.5ex
    \advance\graphtemp by 0.413in
    \rlap{\kern 1.594in\lower\graphtemp\hbox to 0pt{\hss $b$\hss}}%
    \graphtemp=.5ex
    \advance\graphtemp by 1.043in
    \rlap{\kern 1.594in\lower\graphtemp\hbox to 0pt{\hss $b$\hss}}%
    \graphtemp=.5ex
    \advance\graphtemp by 1.043in
    \rlap{\kern 1.988in\lower\graphtemp\hbox to 0pt{\hss $a$\hss}}%
    \graphtemp=.5ex
    \advance\graphtemp by 1.043in
    \rlap{\kern 2.382in\lower\graphtemp\hbox to 0pt{\hss $a$\hss}}%
    \graphtemp=.5ex
    \advance\graphtemp by 1.043in
    \rlap{\kern 2.657in\lower\graphtemp\hbox to 0pt{\hss $\cdots$\hss}}%
    \graphtemp=.5ex
    \advance\graphtemp by 1.437in
    \rlap{\kern 1.988in\lower\graphtemp\hbox to 0pt{\hss $b$\hss}}%
    \graphtemp=.5ex
    \advance\graphtemp by 1.437in
    \rlap{\kern 2.382in\lower\graphtemp\hbox to 0pt{\hss $a$\hss}}%
    \graphtemp=.5ex
    \advance\graphtemp by 1.831in
    \rlap{\kern 2.382in\lower\graphtemp\hbox to 0pt{\hss $b$\hss}}%
\pdfliteral{
q [] 0 d 1 J 1 j
0.576 w
0.072 w
q 0 g
150.624 -23.976 m
148.824 -31.176 l
147.024 -23.976 l
150.624 -23.976 l
B Q
0.576 w
148.824 -5.688 m
148.824 -23.976 l
S
0.072 w
q 0 g
128.808 -18.36 m
121.68 -20.52 l
126.936 -15.264 l
128.808 -18.36 l
B Q
0.576 w
147.6 -4.968 m
127.872 -16.848 l
S
0.072 w
q 0 g
122.256 -35.352 m
120.456 -42.552 l
118.656 -35.352 l
122.256 -35.352 l
B Q
0.576 w
120.456 -22.68 m
120.456 -35.352 l
S
0.072 w
q 0 g
127.872 -56.16 m
121.464 -59.976 l
125.28 -53.568 l
127.872 -56.16 l
B Q
0.576 w
147.816 -33.624 m
126.576 -54.864 l
S
0.072 w
q 0 g
150.624 -52.344 m
148.824 -59.544 l
147.024 -52.344 l
150.624 -52.344 l
B Q
0.576 w
148.824 -33.984 m
148.824 -52.344 l
S
0.072 w
q 0 g
172.368 -53.568 m
176.184 -59.976 l
169.776 -56.16 l
172.368 -53.568 l
B Q
0.576 w
149.832 -33.624 m
171.072 -54.864 l
S
0.072 w
q 0 g
122.256 -80.64 m
120.456 -87.84 l
118.656 -80.64 l
122.256 -80.64 l
B Q
0.576 w
120.456 -62.352 m
120.456 -80.64 l
S
0.072 w
q 0 g
150.624 -80.64 m
148.824 -87.84 l
147.024 -80.64 l
150.624 -80.64 l
B Q
0.576 w
148.824 -62.352 m
148.824 -80.64 l
S
0.072 w
q 0 g
178.992 -80.64 m
177.192 -87.84 l
175.392 -80.64 l
178.992 -80.64 l
B Q
0.576 w
177.192 -62.352 m
177.192 -80.64 l
S
0.072 w
q 0 g
150.624 -109.008 m
148.824 -116.208 l
147.024 -109.008 l
150.624 -109.008 l
B Q
0.576 w
148.824 -90.72 m
148.824 -109.008 l
S
0.072 w
q 0 g
178.992 -109.008 m
177.192 -116.208 l
175.392 -109.008 l
178.992 -109.008 l
B Q
0.576 w
177.192 -90.72 m
177.192 -109.008 l
S
0.072 w
q 0 g
178.992 -137.376 m
177.192 -144.576 l
175.392 -137.376 l
178.992 -137.376 l
B Q
0.576 w
177.192 -119.088 m
177.192 -137.376 l
S
Q
}%
    \graphtemp=.5ex
    \advance\graphtemp by 2.224in
    \rlap{\kern 0.571in\lower\graphtemp\hbox to 0pt{\hss $P$\hss}}%
\pdfliteral{
q [] 0 d 1 J 1 j
0.576 w
42.552 -32.616 m
42.552 -33.41129 41.90729 -34.056 41.112 -34.056 c
40.31671 -34.056 39.672 -33.41129 39.672 -32.616 c
39.672 -31.82071 40.31671 -31.176 41.112 -31.176 c
41.90729 -31.176 42.552 -31.82071 42.552 -32.616 c
S
14.184 -60.912 m
14.184 -61.70729 13.53929 -62.352 12.744 -62.352 c
11.94871 -62.352 11.304 -61.70729 11.304 -60.912 c
11.304 -60.11671 11.94871 -59.472 12.744 -59.472 c
13.53929 -59.472 14.184 -60.11671 14.184 -60.912 c
S
42.552 -60.912 m
42.552 -61.70729 41.90729 -62.352 41.112 -62.352 c
40.31671 -62.352 39.672 -61.70729 39.672 -60.912 c
39.672 -60.11671 40.31671 -59.472 41.112 -59.472 c
41.90729 -59.472 42.552 -60.11671 42.552 -60.912 c
S
70.92 -60.912 m
70.92 -61.70729 70.27529 -62.352 69.48 -62.352 c
68.68471 -62.352 68.04 -61.70729 68.04 -60.912 c
68.04 -60.11671 68.68471 -59.472 69.48 -59.472 c
70.27529 -59.472 70.92 -60.11671 70.92 -60.912 c
S
14.184 -89.28 m
14.184 -90.07529 13.53929 -90.72 12.744 -90.72 c
11.94871 -90.72 11.304 -90.07529 11.304 -89.28 c
11.304 -88.48471 11.94871 -87.84 12.744 -87.84 c
13.53929 -87.84 14.184 -88.48471 14.184 -89.28 c
S
42.552 -89.28 m
42.552 -90.07529 41.90729 -90.72 41.112 -90.72 c
40.31671 -90.72 39.672 -90.07529 39.672 -89.28 c
39.672 -88.48471 40.31671 -87.84 41.112 -87.84 c
41.90729 -87.84 42.552 -88.48471 42.552 -89.28 c
S
70.92 -89.28 m
70.92 -90.07529 70.27529 -90.72 69.48 -90.72 c
68.68471 -90.72 68.04 -90.07529 68.04 -89.28 c
68.04 -88.48471 68.68471 -87.84 69.48 -87.84 c
70.27529 -87.84 70.92 -88.48471 70.92 -89.28 c
S
42.552 -117.648 m
42.552 -118.44329 41.90729 -119.088 41.112 -119.088 c
40.31671 -119.088 39.672 -118.44329 39.672 -117.648 c
39.672 -116.85271 40.31671 -116.208 41.112 -116.208 c
41.90729 -116.208 42.552 -116.85271 42.552 -117.648 c
S
70.92 -117.648 m
70.92 -118.44329 70.27529 -119.088 69.48 -119.088 c
68.68471 -119.088 68.04 -118.44329 68.04 -117.648 c
68.04 -116.85271 68.68471 -116.208 69.48 -116.208 c
70.27529 -116.208 70.92 -116.85271 70.92 -117.648 c
S
70.92 -146.016 m
70.92 -146.81129 70.27529 -147.456 69.48 -147.456 c
68.68471 -147.456 68.04 -146.81129 68.04 -146.016 c
68.04 -145.22071 68.68471 -144.576 69.48 -144.576 c
70.27529 -144.576 70.92 -145.22071 70.92 -146.016 c
S
Q
}%
    \graphtemp=.5ex
    \advance\graphtemp by 0.634in
    \rlap{\kern 0.331in\lower\graphtemp\hbox to 0pt{\hss $a$\hss}}%
    \graphtemp=.5ex
    \advance\graphtemp by 0.650in
    \rlap{\kern 0.839in\lower\graphtemp\hbox to 0pt{\hss $a$\hss}}%
    \graphtemp=.5ex
    \advance\graphtemp by 0.650in
    \rlap{\kern 0.512in\lower\graphtemp\hbox to 0pt{\hss $a$\hss}}%
    \graphtemp=.5ex
    \advance\graphtemp by 1.043in
    \rlap{\kern 0.098in\lower\graphtemp\hbox to 0pt{\hss $b$\hss}}%
    \graphtemp=.5ex
    \advance\graphtemp by 1.043in
    \rlap{\kern 0.492in\lower\graphtemp\hbox to 0pt{\hss $a$\hss}}%
    \graphtemp=.5ex
    \advance\graphtemp by 1.043in
    \rlap{\kern 0.886in\lower\graphtemp\hbox to 0pt{\hss $a$\hss}}%
    \graphtemp=.5ex
    \advance\graphtemp by 1.043in
    \rlap{\kern 1.161in\lower\graphtemp\hbox to 0pt{\hss $\cdots$\hss}}%
    \graphtemp=.5ex
    \advance\graphtemp by 1.437in
    \rlap{\kern 0.492in\lower\graphtemp\hbox to 0pt{\hss $b$\hss}}%
    \graphtemp=.5ex
    \advance\graphtemp by 1.437in
    \rlap{\kern 0.886in\lower\graphtemp\hbox to 0pt{\hss $a$\hss}}%
    \graphtemp=.5ex
    \advance\graphtemp by 1.831in
    \rlap{\kern 0.886in\lower\graphtemp\hbox to 0pt{\hss $b$\hss}}%
\pdfliteral{
q [] 0 d 1 J 1 j
0.576 w
0.072 w
q 0 g
20.088 -56.16 m
13.752 -59.976 l
17.568 -53.568 l
20.088 -56.16 l
B Q
0.576 w
40.104 -33.624 m
18.864 -54.864 l
S
0.072 w
q 0 g
42.912 -52.344 m
41.112 -59.544 l
39.312 -52.344 l
42.912 -52.344 l
B Q
0.576 w
41.112 -33.984 m
41.112 -52.344 l
S
0.072 w
q 0 g
64.656 -53.568 m
68.472 -59.976 l
62.064 -56.16 l
64.656 -53.568 l
B Q
0.576 w
42.12 -33.624 m
63.36 -54.864 l
S
0.072 w
q 0 g
14.544 -80.64 m
12.744 -87.84 l
10.944 -80.64 l
14.544 -80.64 l
B Q
0.576 w
12.744 -62.352 m
12.744 -80.64 l
S
0.072 w
q 0 g
42.912 -80.64 m
41.112 -87.84 l
39.312 -80.64 l
42.912 -80.64 l
B Q
0.576 w
41.112 -62.352 m
41.112 -80.64 l
S
0.072 w
q 0 g
71.28 -80.64 m
69.48 -87.84 l
67.68 -80.64 l
71.28 -80.64 l
B Q
0.576 w
69.48 -62.352 m
69.48 -80.64 l
S
0.072 w
q 0 g
42.912 -109.008 m
41.112 -116.208 l
39.312 -109.008 l
42.912 -109.008 l
B Q
0.576 w
41.112 -90.72 m
41.112 -109.008 l
S
0.072 w
q 0 g
71.28 -109.008 m
69.48 -116.208 l
67.68 -109.008 l
71.28 -109.008 l
B Q
0.576 w
69.48 -90.72 m
69.48 -109.008 l
S
0.072 w
q 0 g
71.28 -137.376 m
69.48 -144.576 l
67.68 -137.376 l
71.28 -137.376 l
B Q
0.576 w
69.48 -119.088 m
69.48 -137.376 l
S
Q
}%
    \graphtemp=.5ex
    \advance\graphtemp by 2.224in
    \rlap{\kern 5.374in\lower\graphtemp\hbox to 0pt{\hss $Q'$\hss}}%
\pdfliteral{
q [] 0 d 1 J 1 j
0.576 w
388.368 -4.248 m
388.368 -5.04329 387.72329 -5.688 386.928 -5.688 c
386.13271 -5.688 385.488 -5.04329 385.488 -4.248 c
385.488 -3.45271 386.13271 -2.808 386.928 -2.808 c
387.72329 -2.808 388.368 -3.45271 388.368 -4.248 c
S
360 -21.24 m
360 -22.03529 359.35529 -22.68 358.56 -22.68 c
357.76471 -22.68 357.12 -22.03529 357.12 -21.24 c
357.12 -20.44471 357.76471 -19.8 358.56 -19.8 c
359.35529 -19.8 360 -20.44471 360 -21.24 c
S
360 -43.92 m
360 -44.71529 359.35529 -45.36 358.56 -45.36 c
357.76471 -45.36 357.12 -44.71529 357.12 -43.92 c
357.12 -43.12471 357.76471 -42.48 358.56 -42.48 c
359.35529 -42.48 360 -43.12471 360 -43.92 c
S
388.368 -32.616 m
388.368 -33.41129 387.72329 -34.056 386.928 -34.056 c
386.13271 -34.056 385.488 -33.41129 385.488 -32.616 c
385.488 -31.82071 386.13271 -31.176 386.928 -31.176 c
387.72329 -31.176 388.368 -31.82071 388.368 -32.616 c
S
360 -60.912 m
360 -61.70729 359.35529 -62.352 358.56 -62.352 c
357.76471 -62.352 357.12 -61.70729 357.12 -60.912 c
357.12 -60.11671 357.76471 -59.472 358.56 -59.472 c
359.35529 -59.472 360 -60.11671 360 -60.912 c
S
388.368 -60.912 m
388.368 -61.70729 387.72329 -62.352 386.928 -62.352 c
386.13271 -62.352 385.488 -61.70729 385.488 -60.912 c
385.488 -60.11671 386.13271 -59.472 386.928 -59.472 c
387.72329 -59.472 388.368 -60.11671 388.368 -60.912 c
S
416.736 -60.912 m
416.736 -61.70729 416.09129 -62.352 415.296 -62.352 c
414.50071 -62.352 413.856 -61.70729 413.856 -60.912 c
413.856 -60.11671 414.50071 -59.472 415.296 -59.472 c
416.09129 -59.472 416.736 -60.11671 416.736 -60.912 c
S
453.6 -60.912 m
453.6 -61.70729 452.95529 -62.352 452.16 -62.352 c
451.36471 -62.352 450.72 -61.70729 450.72 -60.912 c
450.72 -60.11671 451.36471 -59.472 452.16 -59.472 c
452.95529 -59.472 453.6 -60.11671 453.6 -60.912 c
S
360 -89.28 m
360 -90.07529 359.35529 -90.72 358.56 -90.72 c
357.76471 -90.72 357.12 -90.07529 357.12 -89.28 c
357.12 -88.48471 357.76471 -87.84 358.56 -87.84 c
359.35529 -87.84 360 -88.48471 360 -89.28 c
S
388.368 -89.28 m
388.368 -90.07529 387.72329 -90.72 386.928 -90.72 c
386.13271 -90.72 385.488 -90.07529 385.488 -89.28 c
385.488 -88.48471 386.13271 -87.84 386.928 -87.84 c
387.72329 -87.84 388.368 -88.48471 388.368 -89.28 c
S
416.736 -89.28 m
416.736 -90.07529 416.09129 -90.72 415.296 -90.72 c
414.50071 -90.72 413.856 -90.07529 413.856 -89.28 c
413.856 -88.48471 414.50071 -87.84 415.296 -87.84 c
416.09129 -87.84 416.736 -88.48471 416.736 -89.28 c
S
453.6 -89.28 m
453.6 -90.07529 452.95529 -90.72 452.16 -90.72 c
451.36471 -90.72 450.72 -90.07529 450.72 -89.28 c
450.72 -88.48471 451.36471 -87.84 452.16 -87.84 c
452.95529 -87.84 453.6 -88.48471 453.6 -89.28 c
S
388.368 -117.648 m
388.368 -118.44329 387.72329 -119.088 386.928 -119.088 c
386.13271 -119.088 385.488 -118.44329 385.488 -117.648 c
385.488 -116.85271 386.13271 -116.208 386.928 -116.208 c
387.72329 -116.208 388.368 -116.85271 388.368 -117.648 c
S
416.736 -117.648 m
416.736 -118.44329 416.09129 -119.088 415.296 -119.088 c
414.50071 -119.088 413.856 -118.44329 413.856 -117.648 c
413.856 -116.85271 414.50071 -116.208 415.296 -116.208 c
416.09129 -116.208 416.736 -116.85271 416.736 -117.648 c
S
453.6 -117.648 m
453.6 -118.44329 452.95529 -119.088 452.16 -119.088 c
451.36471 -119.088 450.72 -118.44329 450.72 -117.648 c
450.72 -116.85271 451.36471 -116.208 452.16 -116.208 c
452.95529 -116.208 453.6 -116.85271 453.6 -117.648 c
S
416.736 -146.016 m
416.736 -146.81129 416.09129 -147.456 415.296 -147.456 c
414.50071 -147.456 413.856 -146.81129 413.856 -146.016 c
413.856 -145.22071 414.50071 -144.576 415.296 -144.576 c
416.09129 -144.576 416.736 -145.22071 416.736 -146.016 c
S
Q
}%
    \graphtemp=.5ex
    \advance\graphtemp by 0.236in
    \rlap{\kern 5.315in\lower\graphtemp\hbox to 0pt{\hss $a$\hss}}%
    \graphtemp=.5ex
    \advance\graphtemp by 0.098in
    \rlap{\kern 5.197in\lower\graphtemp\hbox to 0pt{\hss $a$\hss}}%
    \graphtemp=.5ex
    \advance\graphtemp by 0.650in
    \rlap{\kern 5.642in\lower\graphtemp\hbox to 0pt{\hss $a$\hss}}%
    \graphtemp=.5ex
    \advance\graphtemp by 0.630in
    \rlap{\kern 5.134in\lower\graphtemp\hbox to 0pt{\hss $a$\hss}}%
    \graphtemp=.5ex
    \advance\graphtemp by 0.650in
    \rlap{\kern 5.315in\lower\graphtemp\hbox to 0pt{\hss $a$\hss}}%
    \graphtemp=.5ex
    \advance\graphtemp by 1.043in
    \rlap{\kern 4.902in\lower\graphtemp\hbox to 0pt{\hss $b$\hss}}%
    \graphtemp=.5ex
    \advance\graphtemp by 0.453in
    \rlap{\kern 4.902in\lower\graphtemp\hbox to 0pt{\hss $b$\hss}}%
    \graphtemp=.5ex
    \advance\graphtemp by 1.043in
    \rlap{\kern 5.295in\lower\graphtemp\hbox to 0pt{\hss $a$\hss}}%
    \graphtemp=.5ex
    \advance\graphtemp by 1.043in
    \rlap{\kern 5.689in\lower\graphtemp\hbox to 0pt{\hss $a$\hss}}%
    \graphtemp=.5ex
    \advance\graphtemp by 1.004in
    \rlap{\kern 5.965in\lower\graphtemp\hbox to 0pt{\hss $\cdots$\hss}}%
    \graphtemp=.5ex
    \advance\graphtemp by 0.650in
    \rlap{\kern 6.004in\lower\graphtemp\hbox to 0pt{\hss $a$\hss}}%
    \graphtemp=.5ex
    \advance\graphtemp by 1.437in
    \rlap{\kern 5.295in\lower\graphtemp\hbox to 0pt{\hss $b$\hss}}%
    \graphtemp=.5ex
    \advance\graphtemp by 1.437in
    \rlap{\kern 5.689in\lower\graphtemp\hbox to 0pt{\hss $a$\hss}}%
    \graphtemp=.5ex
    \advance\graphtemp by 1.043in
    \rlap{\kern 6.220in\lower\graphtemp\hbox to 0pt{\hss $a$\hss}}%
    \graphtemp=.5ex
    \advance\graphtemp by 1.831in
    \rlap{\kern 5.689in\lower\graphtemp\hbox to 0pt{\hss $b$\hss}}%
\pdfliteral{
q [] 0 d 1 J 1 j
0.576 w
0.072 w
q 0 g
388.728 -23.976 m
386.928 -31.176 l
385.128 -23.976 l
388.728 -23.976 l
B Q
0.576 w
386.928 -5.688 m
386.928 -23.976 l
S
0.072 w
q 0 g
366.912 -18.36 m
359.784 -20.52 l
365.04 -15.264 l
366.912 -18.36 l
B Q
0.576 w
385.704 -4.968 m
365.976 -16.848 l
S
0.072 w
q 0 g
360.36 -35.352 m
358.56 -42.552 l
356.76 -35.352 l
360.36 -35.352 l
B Q
0.576 w
358.56 -22.68 m
358.56 -35.352 l
S
0.072 w
q 0 g
365.976 -56.16 m
359.568 -59.976 l
363.384 -53.568 l
365.976 -56.16 l
B Q
0.576 w
385.92 -33.624 m
364.68 -54.864 l
S
0.072 w
q 0 g
388.728 -52.344 m
386.928 -59.544 l
385.128 -52.344 l
388.728 -52.344 l
B Q
0.576 w
386.928 -33.984 m
386.928 -52.344 l
S
0.072 w
q 0 g
410.472 -53.568 m
414.288 -59.976 l
407.88 -56.16 l
410.472 -53.568 l
B Q
0.576 w
387.936 -33.624 m
409.176 -54.864 l
S
0.072 w
q 0 g
444.96 -55.872 m
450.792 -60.408 l
443.52 -59.184 l
444.96 -55.872 l
B Q
0.576 w
388.224 -33.192 m
444.24 -57.528 l
S
0.072 w
q 0 g
360.36 -80.64 m
358.56 -87.84 l
356.76 -80.64 l
360.36 -80.64 l
B Q
0.576 w
358.56 -62.352 m
358.56 -80.64 l
S
0.072 w
q 0 g
388.728 -80.64 m
386.928 -87.84 l
385.128 -80.64 l
388.728 -80.64 l
B Q
0.576 w
386.928 -62.352 m
386.928 -80.64 l
S
0.072 w
q 0 g
417.096 -80.64 m
415.296 -87.84 l
413.496 -80.64 l
417.096 -80.64 l
B Q
0.576 w
415.296 -62.352 m
415.296 -80.64 l
S
0.072 w
q 0 g
453.96 -80.64 m
452.16 -87.84 l
450.36 -80.64 l
453.96 -80.64 l
B Q
0.576 w
452.16 -62.352 m
452.16 -80.64 l
S
0.072 w
q 0 g
388.728 -109.008 m
386.928 -116.208 l
385.128 -109.008 l
388.728 -109.008 l
B Q
0.576 w
386.928 -90.72 m
386.928 -109.008 l
S
0.072 w
q 0 g
417.096 -109.008 m
415.296 -116.208 l
413.496 -109.008 l
417.096 -109.008 l
B Q
0.576 w
415.296 -90.72 m
415.296 -109.008 l
S
0.072 w
q 0 g
417.096 -137.376 m
415.296 -144.576 l
413.496 -137.376 l
417.096 -137.376 l
B Q
0.576 w
415.296 -119.088 m
415.296 -137.376 l
S
0.072 w
q 0 g
453.96 -109.008 m
452.16 -116.208 l
450.36 -109.008 l
453.96 -109.008 l
B Q
0.576 w
q [3.6 3.744] 0 d
452.16 -90.72 m
452.16 -109.008 l
S Q
0.072 w
q 0 g
453.96 -137.376 m
452.16 -144.576 l
450.36 -137.376 l
453.96 -137.376 l
B Q
0.576 w
q [0 3.6576] 0 d
452.16 -119.088 m
452.16 -137.376 l
S Q
Q
}%
    \graphtemp=.5ex
    \advance\graphtemp by 2.224in
    \rlap{\kern 3.602in\lower\graphtemp\hbox to 0pt{\hss $Q$\hss}}%
\pdfliteral{
q [] 0 d 1 J 1 j
0.576 w
260.784 -32.616 m
260.784 -33.41129 260.13929 -34.056 259.344 -34.056 c
258.54871 -34.056 257.904 -33.41129 257.904 -32.616 c
257.904 -31.82071 258.54871 -31.176 259.344 -31.176 c
260.13929 -31.176 260.784 -31.82071 260.784 -32.616 c
S
232.488 -60.912 m
232.488 -61.70729 231.84329 -62.352 231.048 -62.352 c
230.25271 -62.352 229.608 -61.70729 229.608 -60.912 c
229.608 -60.11671 230.25271 -59.472 231.048 -59.472 c
231.84329 -59.472 232.488 -60.11671 232.488 -60.912 c
S
260.784 -60.912 m
260.784 -61.70729 260.13929 -62.352 259.344 -62.352 c
258.54871 -62.352 257.904 -61.70729 257.904 -60.912 c
257.904 -60.11671 258.54871 -59.472 259.344 -59.472 c
260.13929 -59.472 260.784 -60.11671 260.784 -60.912 c
S
289.152 -60.912 m
289.152 -61.70729 288.50729 -62.352 287.712 -62.352 c
286.91671 -62.352 286.272 -61.70729 286.272 -60.912 c
286.272 -60.11671 286.91671 -59.472 287.712 -59.472 c
288.50729 -59.472 289.152 -60.11671 289.152 -60.912 c
S
326.016 -60.912 m
326.016 -61.70729 325.37129 -62.352 324.576 -62.352 c
323.78071 -62.352 323.136 -61.70729 323.136 -60.912 c
323.136 -60.11671 323.78071 -59.472 324.576 -59.472 c
325.37129 -59.472 326.016 -60.11671 326.016 -60.912 c
S
232.488 -89.28 m
232.488 -90.07529 231.84329 -90.72 231.048 -90.72 c
230.25271 -90.72 229.608 -90.07529 229.608 -89.28 c
229.608 -88.48471 230.25271 -87.84 231.048 -87.84 c
231.84329 -87.84 232.488 -88.48471 232.488 -89.28 c
S
260.784 -89.28 m
260.784 -90.07529 260.13929 -90.72 259.344 -90.72 c
258.54871 -90.72 257.904 -90.07529 257.904 -89.28 c
257.904 -88.48471 258.54871 -87.84 259.344 -87.84 c
260.13929 -87.84 260.784 -88.48471 260.784 -89.28 c
S
289.152 -89.28 m
289.152 -90.07529 288.50729 -90.72 287.712 -90.72 c
286.91671 -90.72 286.272 -90.07529 286.272 -89.28 c
286.272 -88.48471 286.91671 -87.84 287.712 -87.84 c
288.50729 -87.84 289.152 -88.48471 289.152 -89.28 c
S
326.016 -89.28 m
326.016 -90.07529 325.37129 -90.72 324.576 -90.72 c
323.78071 -90.72 323.136 -90.07529 323.136 -89.28 c
323.136 -88.48471 323.78071 -87.84 324.576 -87.84 c
325.37129 -87.84 326.016 -88.48471 326.016 -89.28 c
S
260.784 -117.648 m
260.784 -118.44329 260.13929 -119.088 259.344 -119.088 c
258.54871 -119.088 257.904 -118.44329 257.904 -117.648 c
257.904 -116.85271 258.54871 -116.208 259.344 -116.208 c
260.13929 -116.208 260.784 -116.85271 260.784 -117.648 c
S
289.152 -117.648 m
289.152 -118.44329 288.50729 -119.088 287.712 -119.088 c
286.91671 -119.088 286.272 -118.44329 286.272 -117.648 c
286.272 -116.85271 286.91671 -116.208 287.712 -116.208 c
288.50729 -116.208 289.152 -116.85271 289.152 -117.648 c
S
326.016 -117.648 m
326.016 -118.44329 325.37129 -119.088 324.576 -119.088 c
323.78071 -119.088 323.136 -118.44329 323.136 -117.648 c
323.136 -116.85271 323.78071 -116.208 324.576 -116.208 c
325.37129 -116.208 326.016 -116.85271 326.016 -117.648 c
S
289.152 -146.016 m
289.152 -146.81129 288.50729 -147.456 287.712 -147.456 c
286.91671 -147.456 286.272 -146.81129 286.272 -146.016 c
286.272 -145.22071 286.91671 -144.576 287.712 -144.576 c
288.50729 -144.576 289.152 -145.22071 289.152 -146.016 c
S
Q
}%
    \graphtemp=.5ex
    \advance\graphtemp by 0.650in
    \rlap{\kern 3.870in\lower\graphtemp\hbox to 0pt{\hss $a$\hss}}%
    \graphtemp=.5ex
    \advance\graphtemp by 0.630in
    \rlap{\kern 3.362in\lower\graphtemp\hbox to 0pt{\hss $a$\hss}}%
    \graphtemp=.5ex
    \advance\graphtemp by 0.650in
    \rlap{\kern 3.543in\lower\graphtemp\hbox to 0pt{\hss $a$\hss}}%
    \graphtemp=.5ex
    \advance\graphtemp by 1.043in
    \rlap{\kern 3.130in\lower\graphtemp\hbox to 0pt{\hss $b$\hss}}%
    \graphtemp=.5ex
    \advance\graphtemp by 1.043in
    \rlap{\kern 3.524in\lower\graphtemp\hbox to 0pt{\hss $a$\hss}}%
    \graphtemp=.5ex
    \advance\graphtemp by 1.043in
    \rlap{\kern 3.917in\lower\graphtemp\hbox to 0pt{\hss $a$\hss}}%
    \graphtemp=.5ex
    \advance\graphtemp by 1.004in
    \rlap{\kern 4.193in\lower\graphtemp\hbox to 0pt{\hss $\cdots$\hss}}%
    \graphtemp=.5ex
    \advance\graphtemp by 0.650in
    \rlap{\kern 4.232in\lower\graphtemp\hbox to 0pt{\hss $a$\hss}}%
    \graphtemp=.5ex
    \advance\graphtemp by 1.437in
    \rlap{\kern 3.524in\lower\graphtemp\hbox to 0pt{\hss $b$\hss}}%
    \graphtemp=.5ex
    \advance\graphtemp by 1.437in
    \rlap{\kern 3.917in\lower\graphtemp\hbox to 0pt{\hss $a$\hss}}%
    \graphtemp=.5ex
    \advance\graphtemp by 1.043in
    \rlap{\kern 4.449in\lower\graphtemp\hbox to 0pt{\hss $a$\hss}}%
    \graphtemp=.5ex
    \advance\graphtemp by 1.831in
    \rlap{\kern 3.917in\lower\graphtemp\hbox to 0pt{\hss $b$\hss}}%
\pdfliteral{
q [] 0 d 1 J 1 j
0.576 w
0.072 w
q 0 g
238.392 -56.16 m
232.056 -59.976 l
235.872 -53.568 l
238.392 -56.16 l
B Q
0.576 w
258.336 -33.624 m
237.096 -54.864 l
S
0.072 w
q 0 g
261.144 -52.344 m
259.344 -59.544 l
257.544 -52.344 l
261.144 -52.344 l
B Q
0.576 w
259.344 -33.984 m
259.344 -52.344 l
S
0.072 w
q 0 g
282.888 -53.568 m
286.704 -59.976 l
280.368 -56.16 l
282.888 -53.568 l
B Q
0.576 w
260.352 -33.624 m
281.592 -54.864 l
S
0.072 w
q 0 g
317.376 -55.872 m
323.28 -60.408 l
315.936 -59.184 l
317.376 -55.872 l
B Q
0.576 w
260.64 -33.192 m
316.656 -57.528 l
S
0.072 w
q 0 g
232.848 -80.64 m
231.048 -87.84 l
229.248 -80.64 l
232.848 -80.64 l
B Q
0.576 w
231.048 -62.352 m
231.048 -80.64 l
S
0.072 w
q 0 g
261.144 -80.64 m
259.344 -87.84 l
257.544 -80.64 l
261.144 -80.64 l
B Q
0.576 w
259.344 -62.352 m
259.344 -80.64 l
S
0.072 w
q 0 g
289.512 -80.64 m
287.712 -87.84 l
285.912 -80.64 l
289.512 -80.64 l
B Q
0.576 w
287.712 -62.352 m
287.712 -80.64 l
S
0.072 w
q 0 g
326.376 -80.64 m
324.576 -87.84 l
322.776 -80.64 l
326.376 -80.64 l
B Q
0.576 w
324.576 -62.352 m
324.576 -80.64 l
S
0.072 w
q 0 g
261.144 -109.008 m
259.344 -116.208 l
257.544 -109.008 l
261.144 -109.008 l
B Q
0.576 w
259.344 -90.72 m
259.344 -109.008 l
S
0.072 w
q 0 g
289.512 -109.008 m
287.712 -116.208 l
285.912 -109.008 l
289.512 -109.008 l
B Q
0.576 w
287.712 -90.72 m
287.712 -109.008 l
S
0.072 w
q 0 g
289.512 -137.376 m
287.712 -144.576 l
285.912 -137.376 l
289.512 -137.376 l
B Q
0.576 w
287.712 -119.088 m
287.712 -137.376 l
S
0.072 w
q 0 g
326.376 -109.008 m
324.576 -116.208 l
322.776 -109.008 l
326.376 -109.008 l
B Q
0.576 w
q [3.6 3.744] 0 d
324.576 -90.72 m
324.576 -109.008 l
S Q
0.072 w
q 0 g
326.376 -137.376 m
324.576 -144.576 l
322.776 -137.376 l
326.376 -137.376 l
B Q
0.576 w
q [0 3.6576] 0 d
324.576 -119.088 m
324.576 -137.376 l
S Q
Q
}%
    \hbox{\vrule depth2.224in width0pt height 0pt}%
    \kern 6.417in
  }%
}%

\centerline{\box\graph}
\vspace{2ex}

\noindent
In CCS \cite{Mi90ccs}, process $P$ can be expressed as $P := \sum_{i=1}^\infty a^i.b.{\bf 0}$, where
the abbreviation $a^i.R$ is defined by $a^0.R := R$ and $a^{i+1}.R := a.a^i.R$. This process has
branches featuring any finite positive number of $a$-transitions in succession, and each of those branches
ends with a $b$-transition. Process $Q := P+a^\infty$ is like $P$, except that it also has a
branch with infinitely many $a$-transitions, thus not ending with a $b$. The process $a^\infty$
can be defined as the unique solution of the guarded recursive equation $X = a.X$.
Finally, $P':=a.b+a.P$ and $Q':=a.b + a.Q$.

In all semantics $\equiv$ of \cite{vG01} between trace and
ready trace equivalence one has $P' \equiv P$ and $Q \equiv Q'$.
Thus, the guarded recursive specification $\RS = \{X = a.b + a.X\}$
has both $P$ and $Q$ as solutions.
Hence for all such semantic equivalences that distinguish between $P$ and $Q$ in models of
concurrency where $P$ and $Q$ both exist, RSP fails. This applies to semantics that take
infinite traces (or infinite ready traces) into account. A reason to distinguish between $P$ and
$Q$ is that $P$ is guaranteed to eventually perform a $b$-transition, whereas $Q$ is not.
For linear time semantics that do not take infinite traces into account, and identify $P$ and $Q$,
RSP may hold. But that is best established by means of fixed point arguments \cite{Ros97},
which is out of scope for this paper.

After establishing RSP and RDP for bisimilarity and (ready) simulation equivalence, as well as for
the corresponding preorders, I establish two corollaries, namely (1) that these equivalences and
preorders are full (pre)congruences for guarded recursion, and (2) a sound and ground-complete
axiomatisation of strong bisimilarity for any finitary GSOS language.

\section{Transition system specifications and their meaning}

In this paper $\Var$ and $A$ are two sets of {\em variables} and {\em
actions}. Many concepts that will appear are param-\linebreak[4]eterised by the
choice of $\Var$ and $A$, but as here this choice is fixed, a
corresponding index is suppressed.

\begin{definitionA}{\emph{Terms}}{signature}
A {\em signature} is a set $\Sigma$ of operator symbols $f \not\in \Var$, each equipped with an
\emph{arity} $\ar{f}\in\IN$.\footnote{This work, prior to \Sec{full}, generalises seamlessly to
  operators with infinitely many arguments. Such operators occur, for instance, in
  \cite[Appendix A.2]{BrGH16b}. 
  Hence one may take $\ar{f}$ to be any ordinal.  An operator, like the \emph{summation} or \emph{choice}
  of CCS \cite{Mi90ccs}, that actually takes any \emph{set} of arguments, needs to be simulated by a
  family of operators with a \emph{sequence} of arguments (but yielding the same value upon
  reshuffling of the arguments), one for each cardinality of this set.} The set
$\IT(\Sigma)$ of {\em terms with recursion} over a signature $\Sigma$ is defined
inductively by:
\begin{itemize}
\item $\Var \subseteq \IT(\Sigma)$,
\item if $f \mathbin\in \Sigma$ and $\E_1,...,\E_{\ar{f}} \mathbin\in \IT(\Sigma)$ then
$f(\E_1,...,\E_{\ar{f}}) \mathbin\in \IT(\Sigma)$,
\item If $V_\RS \subseteq \Var$, $~\RS:V_\RS \rightarrow \IT(\Sigma)$ and $X\in V_\RS$,
then $\rec{X|\RS}\in \IT(\Sigma)$.
\end{itemize}
A term $c()$ is abbreviated as $c$.\
A term $\rec{X|\RS}$ as appears in the last clause is a \emph{recursive call}, and
the function $\RS$ therein is called a \emph{recursive specification}\index{recursive specification}.
It is often displayed as $\{X=\RS_X \mid X\in V_\RS\}$.
So $V_\RS$ is the domain of $\RS$ and $\RS_X$ represents $\RS(X)$.
  Each term $\RS_Y$ for $Y\in V_\RS$ counts as a subterm of $\rec{X|\RS}$.
An occurrence of a variable $y$ in a term $\E$ is {\em free} if it does not
occur in a subterm of $\E$ of the form $\rec{X|\RS}$ with $y \in V_\RS$.
Let $\var(\E)$ denote the set of variables occurring free in a
term $\E\in\IT(\Sigma)$, and for $W\subseteq \Var$ let $\IT(\Sigma,W)$ be
the set of terms $\E$ over $\Sigma$ with $\var(\E)\subseteq W$.
$\T(\Sigma):=\IT(\Sigma,\emptyset)$ is the set of \emph{closed} terms over $\Sigma$,
modelling processes.
\end{definitionA}

\begin{example}{recursion}
Let $\Sigma$ contain three unary operators $a.\_$, $b.\_$ and $d.\_$, and an infix-written binary
operator $\|$.
Let $X,Y,z\mathbin\in\Var$. Then $\RS=\{X=(a.X)\|(b.Y),~ Y=(d.Y)\|(X\|z)\}$ is a recursive
specification, so $\rec{X|\RS}\in\IT(\Sigma)$. Since $V_\RS\mathbin=\{X,Y\}$, the only variable that
occurs free in this term is $z$.
\end{example}
As illustrated here, I often choose upper case letters for bound variables (the ones occurring in a set
$V_\RS$) and lower case ones for variables occurring free; this is a convention only.

A recursive specification $\RS$ is meant to denote a $V_\RS$-tuple (in the example above a pair) of
processes that---when filled in for the variables in $V_\RS$---forms a solution to the equations
in $\RS$.\footnote{When $\RS$ contains free variables from a set $W$, this solution is parameterised by the choice of a
valuation of these variables as processes, thereby becoming a $W$-ary function.}
The term $\rec{X|\RS}$ denotes the $X$-component of such a tuple.

\begin{definitionA}{\emph{Substitution}}{substitutions}
A {\em $\Sigma$-substitution} $\sigma$ is a partial function from $\Var$ to
$\IT(\Sigma)$; it is \emph{closed} if it is a total function from $\Var$ to $\T(\Sigma)$.
If $\sigma$ is a substitution and $S$ any syntactic
object, then $S[\sigma]$ denotes the object obtained from $S$ by
replacing, for $x$ in the domain of $\sigma$, every free occurrence of $x$
in $S$ by $\sigma(x)$, while renaming bound variables if necessary to prevent
name-clashes. In that case $S[\sigma]$ is called a {\em substitution instance} of $S$.
A substitution instance $S[\sigma]$ where $\sigma$
is given by $\sigma(x_i)=u_i$ for $i\in I$ is denoted as $S[u_i/x_i]_{i\in I}$,
and for a recursive specification $\RS$, the expression $\rec{\E|\RS}$ abbreviates $\E[\rec{Y|\RS}/Y]_{Y\in V_\RS}$.
\end{definitionA}
\begin{example}{extend}
Extend $\Sigma$ from \ex{recursion} with a constant $c$. Then

$\rec{X|\RS}[b.c/z] = \rec{X|\{X{=}(a.X)\|(b.Y),\;  Y{=}(d.Y)\|(X\|b.c)\}}$,

$\rec{X|\RS}[X/z] = \rec{Z|\{Z{=}(a.Z)\|(b.Y),\;  Y{=}(d.Y)\|(Z\|X)\}}$ and

$\rec{X|\RS}[b.c/Y] = \rec{X|\RS}$.
\end{example}

\noindent
Structural operational semantics \cite{Pl04} defines the meaning of system description languages
whose syntax is given by a signature $\Sigma$. It generates a transition system in which the
states, or \emph{processes}, are the closed terms over $\Sigma$---representing the remaining system
behaviour from that state---and transitions between processes are supplied with labels. The
transitions between processes are obtained from a transition system specification, which consists of
a set of transition rules.

\begin{definitionA}{\emph{Transition system specifications}}{TSS}
Let $\Sigma$ be a signature. A {\em positive $\Sigma$-literal} is an
expression \plat{$\E \goesto{a} \E'$} and a {\em negative $\Sigma$-literal} an
expression \plat{$\E \goesnotto{a}$} with $\E,\E'\in\IT(\Sigma)$ and $a \in A$.
For $\E,\E' \in \IT(\Sigma)$ the literals $\E \goesto{a} \E'$ and $\E \goesnotto{a}$
are said to {\em deny} each other.
A {\em transition rule} over $\Sigma$ is an expression of the form
$\frac{H}{\alpha}$ with $H$ a set of $\Sigma$-literals (the {\em
premises} or {\em antecedents} of the rule) and $\alpha$ a
positive $\Sigma$-literal (the {\em conclusion}).
The terms at the left- and right-hand side of $\alpha$ are
the \emph{source} and \emph{target} of the rule.
A rule \plat{$\frac{H}{\alpha}$} with $H=\emptyset$ is also written $\alpha$.
A literal or transition rule is {\em closed} if it contains no free variables.
A {\em transition system specification (TSS)} is a pair $(\Sigma,\R)$
with $\Sigma$ a signature and $\R$ a set of transition rules over $\Sigma$; it is
{\em positive} if all antecedents of its rules are positive.
\end{definitionA}
The concept of a (positive) TSS presented above was introduced in
{\sc Groote \& Vaandrager} \cite{GrV92}; the negative premises \plat{$\E\goesnotto{a}$}
were added in {\sc Groote} \cite{Gr93}. The notion
generalises the {\em GSOS rule systems} of \cite{BIM95} and constitutes
a formalisation of {\sc Plotkin}'s {\em Structural Operational
Semantics (SOS)} \cite{Pl04} that is sufficiently general to cover
many of its applications.

\newcommand{\en}{\textit{en}}
\newcommand{\sw}{\textit{sw}}
\begin{example}{TSS}
  Let $\Sigma$ feature a constant $0$, unary operators $a.\_$ for $a$ ranging over the set
  $A$ of actions, binary operators $\_ \triangleright \_$ and $\_ \mathbin{\|^{\mbox{}}_S} \_$ for each
  $S \subseteq A$, and unary operators $\en_R(\_)$ for each $R \subseteq A$. Their semantics is
  given by the following set of transition rules:\vspace{-4pt}
  \[\begin{array}{@{}c@{}}\displaystyle
  a.x \goto\al x
  \quad \frac{x \goto{a} x'}{x\mathbin{\|^{\mbox{}}_S} y \goto{a} x'\mathbin{\|^{\mbox{}}_S} y}~{\scriptstyle(a \mathbin{\notin} S)}
  \quad \frac{x \goto{a} x' \quad y \goto{a} y'}{x\mathbin{\|^{\mbox{}}_S} y \goto{a} x'\mathbin{\|^{\mbox{}}_S} y'}~{\scriptstyle(a \mathbin\in S)}
  \quad \frac{y \goto{a} y'}{x\mathbin{\|^{\mbox{}}_S} y \goto{a} x\mathbin{\|^{\mbox{}}_S} y'}~{\scriptstyle(a \mathbin{\notin} S)}
  \quad \frac{\rec{\RS_X|\RS} \goto{a} z}{\rec{X|\RS} \goto{a} z}
  \\[16pt] \displaystyle
  \frac{x \goto{a} x'}{\en_R(x) \goto{a} \en_R(x')}
  \qquad \frac{x \gonotto{a}}{\en_R(x) \goto{a} \en_R(x)}~{\scriptstyle(a \in R)}
  \qquad \frac{x \goto{a} x'}{x\triangleright y \goto{a} x'\triangleright y}~{\scriptstyle(a \neq \sw)}
  \qquad \frac{x \goto{\sw} x'}{x\triangleright y \goto{\sw} y\triangleright x'}
  \end{array}\]
  Here all rules displayed are really rule templates, with one instance for each choice of $a \in A$,
  except when this choice is restricted by the side condition. The operators $0$ and $a.\_$ stem
  from CCS \cite{Mi90ccs} and $\|^{\mbox{}}_S$ from CSP \cite{OH86}. The operators $\en_R$ and
  $\triangleright$ were invented to create this example.
  
  The process $0$ cannot perform any actions, and the process $a.P$ first performs the action $a$
  and then behaves as $P$. The process $P \mathbin{\|^{\mbox{}}_S} Q$ is a partially synchronous parallel
  composition of the processes $P$ and $Q$, where actions $a\notin S$ from $P$ and $Q$ may occur
  independently, whereas actions $a\in S$ can occur only when both $P$ and $Q$ partake in such a
  synchronisation. The process $\en_R(P)$ behaves like $P$, but with the modification that the
  actions from $R$ must be enabled in any state. In a state where $P$ cannot perform such an action,
  $\en_R(P)$ adds it as a self-loop. The purpose of this operator is to take $R$ to be those actions
  that are interpreted as receiving a message, so that the process $\en_R(P)$ can receive a message
  in any state. In case $P$ itself could not receive the message (by having an appropriate outgoing
  transition), then $\en_R(P)$ adds the possibility of receiving and dropping said message.
  As a consequence, whereas in $P \mathbin{\|^{\mbox{}}_S} Q$ the process $P$ can be blocked when it wants to
  ``send'' a message $b \in S$ that $Q$ is not ready to receive, in $P \mathbin{\|^{\mbox{}}_S} \en_S(Q)$ no
  such blocking can occur. Finally, the operator $P \triangleright Q$ lets $P$ work while $Q$ rests,
  until $P$ performs the action $\sw$(itch)$\;\in A$, after which $Q$ works and $P$ rests.
\end{example}

\noindent
The following definition (from \cite{vG93d}) tells when a transition is provable from a
TSS\@. It generalises the standard definition (see e.g.\ \cite{GrV92})
by (also) allowing the derivation of transition rules. The
derivation of a transition \plat{$\E\goesto{a}\E'$} corresponds to the derivation
of the transition rule \plat{$\frac{H}{\E\rule{0pt}{9pt}\goto{a}\E'}$} with $H\mathbin=\emptyset$.
The case $H \mathbin{\neq} \emptyset$ corresponds to the derivation of
\plat{$E\goesto{a}\E'$} under the assumptions $H$.

\begin{definitionA}{\emph{Proof}}{proof}
Let $\TS=(\Sigma,\R)$ be a TSS. A {\em proof} of a transition
rule $\frac{H}{\alpha}$ from $\TS$ is a well-founded, upwardly
branching tree of which the nodes are labelled by $\Sigma$-literals,
such that:
\begin{itemize}
\item the root is labelled by $\alpha$, and
\item if $\beta$ is the label of a node $q$ and $K$ is the set of
labels of the nodes directly above $q$, then
\begin{itemize}
\item either $K=\emptyset$ and $\beta \in H$,
\item or $\frac{K}{\beta}$ is a substitution instance of a rule from $\R$.
\end{itemize}
\end{itemize}
If a proof of $\frac{H}{\alpha}$ from $\TS$ exists, then $\frac{H}{\alpha}$
is {\em provable} from $\TS$, notation $\TS \vdash \frac{H}{\alpha}$.
\end{definitionA}

\begin{example}{provable}
  In \ex{TSS} one can prove that
  \plat{$a.b.0 \mathbin{\|_{\{a,b\}}} \en_{\{a,b\}}(b.a.0) \goto{a} b.0 \mathbin{\|_{\{a,b\}}} \en_{\{a,b\}}(b.a.0)$}
  and
  \plat{$b.0 \mathbin{\|_{\{a,b\}}} \en_{\{a,b\}}(b.a.0) \goto{b} 0 \mathbin{\|_{\{a,b\}}} \en_{\{a,b\}}(a.0)$}.
  No further transitions from these three processes are provable.
  Let $\RS$ be the recursive specification $X = b.0 \mathbin{\|_{b}} \en_{b}(X)$.
  Under the assumption that \plat{$\rec{X|\RS} \goto{b} P$} one can now prove that
  \plat{$\rec{X|\RS} \goto{b} 0 \mathbin{\|_{b}} \en_{b}(P)$}, whereas under the assumption that
  \plat{$\rec{X|\RS} \gonotto{b}$} one proves that
  \plat{$\rec{X|\RS} \goto{b} 0 \mathbin{\|_{b}} \en_{b}(\rec{X|\RS})$}. However, there is no
  outright proof of any transition \plat{$\rec{X|\RS} \goto{b} Q$}.
\end{example}

\noindent
A TSS is meant to specify a transition system in which the transitions are closed positive literals.
A positive TSS specifies a transition relation in a straightforward
way as the set of all provable transitions.%
  \footnote{Readers interested only in the restriction of
  my results to TSSs without negative premises---giving rise to 2-valued transition relations---can
  safely skip the remainder of this section, and identify $p\hoto{a}p'$ with $p\goto{a}p'$.
  In the proof of \lem{guarded} the induction on $\lambda$
  can be skipped, as well as Claims 1, 2 and 3, and the proof proceeds directly by induction on $\pi$.}
But as pointed out in {\sc Groote} \cite{Gr93}, it is not so easy to associate a
transition relation to a TSS with negative premises.
In \cite{vG04} several solutions to this problem were reviewed and evaluated.
Arguably, the best method to assign a meaning to all TSSs is the \emph{well-founded semantics}
of {\sc Van Gelder, Ross \& Schlipf} \cite{GRS91}, which in general yields a
{\em 3-valued transition relation}
$T: {\sf T}(\Sigma) \times A \times {\sf T}(\Sigma)\rightarrow
\{\mbox{\sl present}, \mbox{\sl undetermined}, \mbox{\sl absent}\}$.
I present such a relation as a pair $\rec{CT,PT}$ of 2-valued
transition relations---the sets of \emph{certain} and \emph{possible transitions}---with $CT \subseteq PT$.
When insisting on 2-valued transition relations, the best method is the same,
declaring meaningful only those TSSs whose well-founded semantics is 2-valued, meaning that $CT=PT$.
Such a TSS is called \hypertarget{complete}{\emph{complete}} \cite{vG04}.

Below I follow the presentation from \cite{vG17b}, which was strongly inspired by
earlier accounts \cite{Prz90,BolG96,vG04}. As \df{proof} does not allow the
derivation of negative literals, to arrive at an approximation $AT^+$ of the set of transitions
that are in the transition relation intended by a TSS $\TS$, one could start from an approximation
$AT^-$ of the closed negative literals that ought to be generated, and define $AT^+$ as the
set of closed positive literals provable from $\TS$ under the hypotheses $AT^-$.
Intuitively,
\begin{enumerate}
  \item if $AT^-$ is an under- (resp.\ over-)approximation of the closed negative
    literals that ``really'' hold, then $AT^+$ will be an under- (resp.\ over-)approximation
    of the intended (2-valued) transition relation, and
  \item if $AT^+$ is an under- (resp.\ over-)approximation of the intended transition
    relation, then the set of all closed negative literals that do not deny any literal in
    $AT^+$ is an over- (resp.\ under-)approximation of the closed negative literals
    that agree with the intended transition relation.
\end{enumerate}
Based on this insight, for each ordinal $\lambda$ I define overapproximations $PT^+_\lambda$ and
$PT^-_\lambda$ of the closed\linebreak[4] positive and negative literals that the transition relation intends to
generate, and underapproximations $CT^+_\lambda$ and $CT^-_\lambda$. The approximations get better
with increasing $\lambda$. I start by taking the set of \emph{all} closed negative literals as
$PT^-_0$ and then apply 1 and 2 above to generate these approximations in the order
\(PT^-_\lambda \rightarrow PT^+_\lambda \rightarrow CT^-_\lambda \rightarrow CT^+_\lambda \rightarrow
PT^-_{\lambda+1} \rightarrow \cdots\).

\begin{definitionA}{\emph{Over- and underapproximations
      of transition relations} \cite{vG17b}}{well-founded}
Let $\TS$ be a TSS\@.
For ordinals $\lambda$ the sets $CT^+_\lambda$ and $PT^+_\lambda$
of closed positive literals, and $CT^-_\lambda$, $PT^-_\lambda$
of closed negative literals are defined inductively~by:
\vspace{-3ex}

\noindent
\hfill
\quad\;$PT^-_\lambda$ \begin{tabular}{l}~\\is the set of negative literals\\ that do not
  deny any\\ $\beta\in CT^+_\kappa$ with $\kappa<\lambda$
  \end{tabular}
\hfill\;
$\beta \in PT^+_\lambda$ iff $\TS\vdash \frac{PT^-_\lambda}{\beta}$
\hfill\mbox{}\vspace{1ex}

\hfill
$CT^-_\lambda$ \begin{tabular}{l}is the set of negative literals\\ that do not
  deny any $\beta\in PT^+_\lambda$
  \end{tabular}
\hfill
$\beta \in CT^+_\lambda$ iff $\TS\vdash \frac{CT^-_\lambda}{\beta}$.
\hfill\mbox{}
\end{definitionA}

\begin{lemmaA}{\cite{vG17b}}{approximations}
$CT^-_\kappa \mathbin\subseteq CT^-_\lambda \mathbin\subseteq PT^-_\lambda \mathbin\subseteq PT^-_\kappa$
and $CT^+_\kappa \mathbin\subseteq CT^+_\lambda \mathbin\subseteq PT^+_\lambda \mathbin\subseteq PT^+_\kappa$
for~$\kappa\mathbin<\lambda$.
\end{lemmaA}
\begin{definitionA}{\emph{Well-founded semantics} \cite{vG17b}}{well-founded semantics}
Define $PT^- :=\bigcap_\lambda PT^-_\lambda$, taking the intersection over all ordinals.
Likewise, $PT^+ :=\bigcap_\lambda PT^+_\lambda$,
$CT^- := \bigcup_\lambda CT^-_\lambda$ and $CT^+: =\bigcup_\lambda CT^+_\lambda$.

The 3-valued transition relation $\rec{CT^+,PT^+}$ constitutes the well-founded semantics of a TSS.
\end{definitionA}
Since the closed literals over $\Sigma$ form a proper set, there must be an ordinal $\kappa$ such
that $PT^-_\lambda = PT^-_\kappa$ for all $\lambda\mathbin>\kappa$, and hence also $PT^+_\lambda = PT^+_\kappa$,
$CT^-_\lambda = CT^-_\kappa$ and $CT^+_\lambda= CT^+_\kappa$.
Such an ordinal $\kappa$ is called a \emph{closure ordinal}. Now $PT^-\mathbin{=}PT^-_\kappa$,
$PT^+\mathbin{=}PT^+_\kappa$, $CT^-\mathbin{=}CT^-_\kappa$ and $CT^+\mathbin{=}CT^+_\kappa$.
Trivially, one obtains:

\begin{observation}{well-founded}
\hypertarget{note}{\qquad $CT^- \subseteq PT^-$ \hfill\hfill $CT^+ \subseteq PT^+$ \hfill\mbox{}}\\[1ex]
\mbox{}\hfill
\quad\;$PT^-$ \begin{tabular}{l}is the set of negative literals\\ that do not
  deny any $\beta\in CT^+$
  \end{tabular}
\hfill\;
$\beta \in PT^+$ iff $\TS\vdash \frac{PT^-}{\beta}$
\hfill\mbox{}\vspace{1ex}

\hfill
$CT^-$ \begin{tabular}{l}is the set of negative literals\\ that do not
  deny any $\beta\in PT^+$
  \end{tabular}
\hfill
$\beta \in CT^+$ iff $\TS\vdash \frac{CT^-}{\beta}$.
\hfill\mbox{}
\end{observation}
In \cite{vG17b} it is shown that the above account of the well-founded semantics is consistent with
the one in \cite{vG04}, and thereby with the ones in \cite{BolG96,Prz90,GRS91}.

Below the statements
$P \goto{a}_\lambda Q$, $P \gonotto{a}_\lambda$, $P \hoto{a}_\lambda Q$ and
$P \honotto{a}_\lambda$ will abbreviate $(P \goto{a} Q)\in CT^+_\lambda$, 
$(P \gonotto{a})\mathbin\in CT^-_\lambda$, $(P \goto{a} Q)\in PT^+_\lambda$
and $(P \gonotto{a})\in PT^-_\lambda$, respectively.
Likewise, $P \goto{a} Q$ and $P \hoto{a} Q$ abbreviate $(P \goto{a} Q)\in CT^+$, 
and $(P \goto{a} Q)\in PT^+$.

\begin{example}{solution S}
  For $\RS$ the recursive specification of \ex{provable}, we have
  \plat{$\rec{X|\RS} \gonotto{a}$} for all $a \neq b$,
  \plat{$\rec{X|\RS} \honotto{b}$} and \plat{$\rec{X|\RS} \hoto{b} 0 \mathbin{\|_{b}} \en_{b}(\rec{X|\RS})$}.
  For this application the ordinal 0 is already a closure ordinal.
\end{example}

\section{The ready simulation format with recursion}

\begin{definitionA}{\emph{Ready simulation format with recursion}}{ntyft}
An \emph{ntyft/ntyxt rule} is a transition rule whose source contains at most one operator symbol
and no multiple occurrences of the same variable, and in which the right-hand sides of positive
premises are variables that are all distinct and do not occur in the source. It has
\emph{lookahead} if some variable both occurs in the right-hand side of a premise and in the left-hand
side of a premise.
A TSS is in the \emph{ready simulation format with recursion} if for every recursive
specification $\RS$ and $X \in V_\RS$ it has a rule\vspace{-1ex}
\[\displaystyle\frac{\rec{\RS_X|\RS} \goesto{a} z}{\rec{X|\RS}\goesto{a}z}\]
with $z \in \Var$ and all of its other rules are recursion-free\footnote{i.e.\ without
containing recursive calls at all. I conjecture that this restriction can be lifted at the expense of
extra bookkeeping regarding bound variables. As I am not aware of any applications of such a
generalisation I did not explore this further.} ntyft/ntyxt rules without lookahead.
\end{definitionA}

\begin{example}{RS format}
The TSS of \ex{TSS} is in the GSOS format of \cite{BIM88} (cf.~\df{GSOS}), and thus certainly in the
ready simulation format. To make it a non-GSOS example of a TSS in ready simulation format, add for
instance a ternary operator $h$ with a rule
$\displaystyle\frac{\en_{b}(x \triangleright y) \goto{a} x'}{h(x,y,z) \goto{a} h(z,x',y)}$ for each $a \in A$.
\end{example}

\section{Guardedness}

\begin{definitionA}{Frozen arguments of operators \cite{FGW12}}{liquid/frozen}
  Let $\Sigma$ be a signature and $\Gamma$ a unary predicate on the collection
  $\{(f,i)\mid 1 \leq i \leq \ar{f},~f \in \Sigma \}$ of arguments of operators. If
  $\Gamma(f,i)$, then argument $i$ of $f$ is {\em $\overline\Gamma$-frozen}; otherwise it is
  \emph{$\overline\Gamma$-liquid}. An occurrence of a variable $x$ in a term $\E$ is
  {\em $\overline\Gamma$-frozen} iff $\E$ contains a subterm $f(\E_1,\ldots,\E_{\ar{f}})$ such that the
  occurrence of $x$ is in $\E_i$ for a $\overline\Gamma$-frozen argument $i$ of $f$.
\end{definitionA}
To prepare for the definition of guardedness below, I here use a predicate $\Gamma$ that 
marks the frozen arguments of operators. However, I do not want to break the convention of
\cite{FGW12}, where a predicate $\Gamma$ marks the liquid arguments. This is why I define frozen and
liquid arguments and variable occurrences in terms of the complement $\overline\Gamma$ of my $\Gamma$.

\begin{definition}{respect for guardedness}
  Let $\TS$ be a TSS in the ready simulation format with recursion.
  A predicate $\Gamma$ as above \emph{respects guardedness} for $\TS$ if for each rule $r$ of
  $\TS$ with source $f(x_1,\dots,x_{\ar{f}})$ and $i$ a $\overline\Gamma$-frozen argument of $f$, the
  variable $x_i$ only occurs $\overline\Gamma$-frozen in the left-hand sides of the premises of $r$.
\end{definition}
Note that the union of any collection of predicates $\Gamma$ that respect guardedness for
$\TS$ is itself a predicate that respects guardedness for $\TS$. Hence there exists a largest predicate
$\Gamma_{\!\TS}$ that respects guardedness for $\TS$. Henceforth, \emph{$\TS$-guarded} (or
\emph{guarded} when $\TS$ is understood from context) will mean $\overline{\Gamma_{\!\TS}}$-frozen.

Thus, if $r$ is a rule of $\TS$ with source $f(x_1,\dots,x_{\ar{f}})$ and $i$ is a guarded (=
$\overline{\Gamma_{\!\TS}}$-frozen) argument of $f$, then $x_i$ only occurs guarded (=
$\overline{\Gamma_{\!\TS}}$-frozen) in the left-hand sides of the premises of $r$.  

\begin{definition}{guarded}
Let $\TS$ be a TSS in the ready simulation format with recursion.
An occurrence of a variable $x$ in a term $\E$ is \emph{$\TS$-guarded} if it is $\overline\Gamma$-frozen
for some predicate $\Gamma$ on arguments of operators that respects guardedness for $\TS$.
A term $\E$ is \emph{$\TS$-guarded} if all free occurrences of variables in $\E$ are $\TS$-guarded.
\end{definition}

\begin{example}{guarded}
For the LTS $\TS$ of \ex{RS format}, the predicate $\Gamma_{\!\TS}$ marks as $\overline{\Gamma_{\!\TS}}$-frozen
(1) the argument of the unary action prefixing operators $a.\_$, (2) the second argument of the
binary operator $\triangleright$, (3) the third argument of $h$, and (4) the second argument of $h$.
In case (2) this is because there is no rule with source $x \triangleright y$ in which $y$ occurs in
the left-hand side of a premise. Likewise, for (1), there is no rule with source $a.x$ in which $x$
occurs in a premise---in fact, the only rule with source $a.x$ has no premises
at all. Similarly, for (3), the only rule with source $h(x,y,z)$ lacks a premise containing $z$. Case
(4) is slightly harder, as the only rule with source $h(x,y,z)$ does have a premise containing $y$.
However, in the left-hand side of that premise $y$ merely occurs $\overline{\Gamma_{\!\TS}}$-frozen.
Namely, in the term $\en_{b}(x \triangleright y)$ the operator $\triangleright$ guards the
occurrence of $y$. As a consequence, the term $h(a.x,x,x)$ is $\TS$-guarded, as in it the variable
$x$ occurs $\TS$-guarded only.
\end{example}
The next lemma, that for CCS goes back to \textsc{Milner} \cite{Mi90ccs}, is a crucial tool in proving RSP\@.
When creating or adapting a definition of guardedness, the main consideration is to make sure that
this lemma holds. Below $\vec{P},\vec{Q}\in\T^W$ are $W$-tuples of processes; such tuples are also
substitutions---hence the notation $G[\vec{P}]$.

\begin{lemma}{guarded}
Let $G\in\IT(\Sigma,W)$ be guarded and have free variables from $W\subseteq \Var$ only, and let
$\vec{P},\vec{Q}\in\T^W$.\linebreak
If $G[\vec{P}] \goesto\al R$ with $\al\in A$, then $R$ has the form $G'[\vec{P}]$ for some
$G'\in\IT(\Sigma,W)$. Moreover $G[\vec{Q}] \goesto\al G'[\vec{Q}]$.
If $G[\vec{Q}] \hoto\al R$ with $\al\in A$, then $R$ has the form $G'[\vec{Q}]$ for some
$G'\in\IT(\Sigma,W)$. Moreover $G[\vec{P}] \hoto\al G'[\vec{P}]$.
\end{lemma}
\begin{proof}
  Given $\vec{P},\vec{Q}\in\T^W$, I will show, for each ordinal $\lambda$, that
  for each guarded $G\in\IT(\Sigma,W)$ and $\al\in A$,
  \begin{enumerate}
  \item If $G[\vec{Q}] \honotto\al_\lambda$ then $G[\vec{P}] \honotto\al_\lambda$,
  \item If $G[\vec{Q}] \hoto\al_\lambda R$ then $R$ has the form $G'[\vec{Q}]$ for some
$G'\in\IT(\Sigma,W)$. Moreover $G[\vec{P}] \hoto\al_\lambda G'[\vec{P}]$,
  \item If $G[\vec{P}] \goesnotto\al_\lambda$ then $G[\vec{Q}] \goesnotto\al_\lambda$, and
  \item If $G[\vec{P}] \goesto\al_\lambda R$ then $R$ has the form $G'[\vec{P}]$ for some
$G'\in\IT(\Sigma,W)$. Moreover $G[\vec{Q}] \goesto\al_\lambda G'[\vec{Q}]$.
  \end{enumerate}
The statements of the lemma then follow by taking $\lambda$ to be a closure ordinal.  
I apply induction on $\lambda$.
  \begin{enumerate}
\item Suppose $G[\vec{P}] \honotto\al_\lambda$ does not hold. Then, by \df{well-founded},
  the literal $G[\vec{P}] \honotto\al$ denies a $\beta \in CT^+_\kappa$ with $\kappa < \lambda$.
  The literal $\beta$ must have the form $G[\vec{P}] \goesto\al_\kappa R$.
  By induction, using Claim 4 above, $G[\vec{Q}] \goesto\al_\kappa R'$ for some $R'\in \T(\Sigma)$.
  Consequently, $G[\vec{Q}] \honotto\al_\lambda$ does not hold. 
\item The proof of Claim 2 proceeds exactly like the one of Claim 4 below.
\item Suppose $G[\vec{Q}] \gonotto\al_\lambda$ does not hold. Then, by \df{well-founded},
  the literal $G[\vec{Q}] \gonotto\al$ denies a $\beta \in PT^+_\lambda$.
  The literal $\beta$ must have the form $G[\vec{Q}] \hoto\al_\lambda R$.
  Using Claim 2 above, $G[\vec{Q}] \hoto\al_\lambda R'$ for some $R'\in \T(\Sigma)$.
  Consequently, $G[\vec{P}] \honotto\al_\lambda$ does not hold. 
\item The proof of this claim proceeds by a nested induction on the derivation $\pi$ of
  $G[\vec{P}] \goesto\al R$ from the hypotheses $CT^-_\lambda$, making a case distinction on the
  shape of $G$.
  Trivially, if a closed positive literal is provable from a set of closed literals, all literals
  occurring in the proof can be chosen closed. This is the case because applying the same closed
  substitution to all literals in a proof, preserves the property of it being a proof.
  I thus may assume that all literals in $\pi$ are closed.

The case that $G=x\in \Var$ cannot occur, as $G$ is guarded.

Let $G = f(G_1,\dots,G_{\ar{f}})$ for some $f\in \Sigma$ and $G_1,\dots,G_{\ar{f}}\in \IT(\Sigma,W)$.
Let $r=\frac{H}{f(x_1,\dots,x_{\ar{f}}) \goto\al F}$ be the ntyft/ntyxt-rule without lookahead
and $\sigma: V \rightarrow \T(\Sigma)$ the substitution that are used in the last step of $\pi$.
Here $V$ is the set of free variables occurring in $r$.
Then $F[\sigma]=R$ and $f(x_1,\dots,x_{\ar{f}})[\sigma]=G[\vec{P}]$,
so $\sigma(x_i) = G_i[\vec{P}]$ for $i=1,\dots,\ar{f}$.
Now define the substitution $\rho: V \rightarrow \IT(\Sigma,W)$ by\vspace{-2ex}
\begin{center}
\begin{tabular}{l@{~=~}l@{\quad for }l}
  $\rho(x_i)$ & $G_i$ & $i=1,\dots,\ar{f}$, and\\
  $\rho(y)$ & $\sigma(y)$ & all $y \in V{\setminus}\{x_1,\dots,x_{\ar{f}}\}$.
\end{tabular}
\end{center}
Then $E[\sigma] = E[\rho][\vec{P}]$ for any recursion-free term $E \in \IT(\Sigma,V)$.

Given a positive premise $E_y \goesto{c} y$ in $H$, one finds that
$E_y[\rho][\vec{P}] = E_y[\sigma] \goesto{c}_\lambda \sigma(y)$, and
the proof of this transition is a subproof of $\pi$.
Consequently, I may apply the induction hypothesis on this transition, provided I can show that 
$E_y[\rho]$ is guarded.

Any free occurrence of a variable $y$ in $E_y[\rho]$ must lay in a term $G_i$, for some
$i \in \{1,\dots,\ar{f}\}$, that is substituted by $\rho$ for an occurrence of $x_i$ in $E_y$.
In case $x_i$ is a guarded argument of $f$, each occurrence of $x_i$ in $E_y$ is guarded, and hence
the occurrence of $y$ is guarded in $E_y[\rho]$ as well.
In case $x_i$ is an unguarded argument of $f$, as this occurrence of $y$ must be guarded in $G$,
it is guarded in $G_i$, and thus also in $E_y[\rho]$.

So $E_y[\rho]$ is guarded indeed. By induction, $\sigma(y)$ has the form $E'_y[\vec{P}]$ for some
$E'_y\in\IT(\Sigma,W)$. Moreover \plat{$E_y[\rho][\vec{Q}] \goesto{c}_\lambda E'_y[\vec{Q}]$}.
Let $V_{\rm RHS}$ be the set of variables occurring as right-hand sides of premises in $H$,
and define the substitution $\overline\rho: V \rightarrow \IT(\Sigma,W)$ by
\begin{center}
\begin{tabular}{l@{~=~}l@{\quad for }l}
  $\overline\rho(x_i)$ & $G_i$ & $i=1,\dots,\ar{f}$,\\
  $\overline\rho(y)$ & $E'_y$ &  $y \in V_{\rm RHS}$, and\\
  $\overline\rho(y)$ & $\sigma(y)$ & all $y \in V{\setminus}(\{x_1,\dots,x_{\ar{f}}\}\cup V_{\rm RHS})$.
\end{tabular}
\end{center}
Then $E[\sigma] = E[\overline\rho][\vec{P}]$ for any recursion-free term $E \in \IT(\Sigma,V)$.
Thus $R = F[\sigma] = F[\overline\rho][\vec{P}]$, which is the first statement that needed to be proved.

Towards the second, let $\overline\sigma: V \rightarrow \T(\Sigma)$ be the substitution given by
\begin{center}
\begin{tabular}{l@{~=~}l@{\quad for }l}
  $\overline\sigma(x_i)$ & $G_i[\vec{Q}]$ & $i=1,\dots,\ar{f}$,\\
  $\overline\sigma(y)$ & $E'_y[\vec{Q}]$ &  $y \in V_{\rm RHS}$, and\\
  $\overline\sigma(y)$ & $\sigma(y)$ & all $y \in V{\setminus}(\{x_1,\dots,x_{\ar{f}}\}\cup V_{\rm RHS})$.
\end{tabular}
\end{center}
Then $E[\overline\sigma] = E[\overline\rho][\vec{Q}]$ for any recursion-free term $E \in \IT(\Sigma,V)$.
For each premise $E_y \goesto{c} y$ in $H$ one has $E_y[\rho] = E_y[\overline\rho]$ since the
rule $r$ has no lookahead. Thus \[E_y[\overline\sigma] = E_y[\overline\rho][\vec{Q}] =
E_y[\rho][\vec{Q}] \goesto{c}_\lambda E'_y[\vec{Q}] = y[\overline\sigma].\]

For each premise $E \goesnotto{c}$ in $H$ one has
$E[\rho][\vec{P}] = E[\sigma] \goesnotto{c}_\lambda$, and $E[\rho]$ is guarded by the same
reasoning as above. By Claim 3 above, $E[\rho][\vec{Q}]\goesnotto{c}_\lambda$.
Again $E[\rho] = E[\overline\rho]$, as $r$ has no lookahead, so 
\(E[\overline\sigma] = E[\overline\rho][\vec{Q}] = E[\rho][\vec{Q}] \goesnotto{c}_\lambda\).
Applying rule $r$, it follows that
\[G[\vec{Q}] = f(x_1,\dots,x_{\ar{f}})[\overline\sigma] \goesto\al_\lambda F[\overline\sigma]
= F[\overline\rho][\vec{Q}],\] which had to be shown.\pagebreak[3]

Finally, let $G = \rec{X|\RS}$, so that $G[\vec{P}] = \rec{X|\RS[\vec{P}^\dagger]}$, where
$\vec{P}^\dagger$ is the $W {\setminus} V_\RS$-tuple that is left of $\vec{P}$ after deleting the
$Y$-components for $Y\in V_\RS$.
The transition $\rec{\RS_X[\vec{P}^\dagger]|\RS[\vec{P}^\dagger]} \goesto\al_\lambda R$ is derivable through a
subderivation of the one for $\rec{X|\RS[\vec{P}^\dagger]}\goesto\al_\lambda R$.
Moreover, $\rec{\RS_X[\vec{P}^\dagger]|\RS[\vec{P}^\dagger]} = \rec{\RS_X|\RS}[\vec{P}]$.
So by induction, $R$ has the form $E[\vec{P}]$ for some $E\mathbin\in\IT(\Sigma,W)$,
and $\rec{\RS_X|\RS}[\vec{Q}] \goesto\al_\lambda E[\vec{Q}]$. 
Since $\rec{\RS_X|\RS}[\vec{Q}] = \rec{\RS_X[\vec{Q}^\dagger]|\RS[\vec{Q}^\dagger]}$, it follows that 
$G[\vec{Q}] = \rec{X|\RS}[\vec{Q}]= \rec{X|\RS[\vec{Q}^\dagger]}\goesto\al_\lambda G'[\vec{Q}]$. 
\qed
  \end{enumerate}
\end{proof}
\vspace{1ex}

\begin{definitionA}{Guarded recursive specifications}{guarded RS}
A recursive specification $\RS$ is \emph{manifestly guarded} if for each $X,Y\in V_\RS$ all
occurrences of the variable $Y$ in the term $\RS_X$ are guarded.
It is \emph{guarded} if it can be converted into a manifestly guarded recursive specification by
repeated substitution of expressions $\RS_Y$ for variables $Y\in V_\RS$ occurring in expressions
$\RS_Z$ for $Z\in V_\RS$.
\end{definitionA}

\begin{example}{manifestly}
  The recursive specification $\{ X = h(a.b.\sw.X, X, X) \}$ is manifestly guarded.
  The specification $\RS := \{ X = Y \triangleright X, ~ Y = (a.b.\sw.X) \triangleright Y \}$ is not
  manifestly guarded, since $Y$ occurs unguarded in $\RS_X$. However,  $\RS$ is guarded, because
  substituting $\RS_Y =  (a.b.\sw.X) \triangleright Y$ for $Y$ in $\RS_X = Y \triangleright X$
  yields the recursive specification
  $\{ X = ((a.b.\sw.X) \triangleright Y) \triangleright X, ~ Y = (a.b.\sw.X) \triangleright Y \}$,
  which \emph{is} manifestly guarded.
\end{example}

\section{The bisimulation preorder}

Traditionally \cite{Mi90ccs}, bisimilarity is defined on 2-valued transition systems only, whereas
the structural operational semantics of a language specified by a TSS can be 3-valued. Rather than
limit my results to languages specified by \hyperlink{complete}{complete TSSs}, I follow
\cite{vG17b} in using an extension of the notion of bisimilarity to 3-valued transition
systems. This extension stems from \cite{LT88} under the name \emph{modal refinement}; there
3-valued transition systems are called \emph{modal transition systems}.

\hypertarget{bisimilarity}{
\begin{definitionA}{Bisimilarity}{refinement}
Let $\TS = (\Sigma,\R)$ be a TSS\@.
A \emph{bisimulation} $\B$ is a binary relation on $\T(\Sigma)$ such that, for
$P,Q\in\T(\Sigma)$ and $a\in A$,
\begin{itemize}
\item if $P\B Q$ and $P \goto{a} P'$, then there is a $Q'$ with $Q \goto{a} Q'$ and $P' \B Q'$,
\item if $P\B Q$ and $Q \hoto{a} Q'$, then there is a $P'$ with $P \hoto{a} P'$ and $P' \B Q'$.
\end{itemize}
A process $Q\mathbin\in \T(\Sigma)$ is a \emph{modal refinement} of $P\mathbin\in \T(\Sigma)$, notation
$P \sqsubseteq_B Q$, if there exists a bisimulation $\B$ with $P \B Q$.
I call $\sqsubseteq_B$ the \emph{bisimulation preorder}, or \emph{bisimilarity}.
The \hyperlink{kernel}{kernel} of $\sqsubseteq$, given by
${\equiv_B} := {\sqsubseteq_B} \cap {\sqsupseteq_B}$, is \emph{bisimulation equivalence}.
\end{definitionA}}
Clearly, modal refinement is reflexive and transitive, and hence a preorder.
The underlying idea is that a process $P$ with a 3-valued transition relation $\rec{CT,PT}$ is a
\emph{specification} of a process with a 2-valued transition relation, in which the presence or
absence of certain transitions is left open. $CT$ contains the transitions that are \emph{required}
by the specification, and $PT$ the ones that are \emph{allowed}. If $P \sqsubseteq_B Q$, then $Q$
may be closer to the eventual implementation, in the sense that some of the undetermined transitions
have been resolved to present or absent. The requirements of \df{refinement} now say that any
transition that is required by $P$ should be (matched by a transition) required by $Q$, whereas any
transition allowed by $Q$, should certainly be (matched by a transition) allowed by $P$.

In case $P$ and $Q$ are 2-valued (i.e.~\emph{implementations}) the modal refinement relation is
just the traditional notion of bisimilarity \cite{Mi90ccs} (and thus symmetric).

\begin{theoremA}{Lean precongruence \cite{vG17b}}{congruence}
  Let $\TS = (\Sigma,\R)$ be a TSS in the ready simulation format with recursion.
  Let $E\in\IT(\Sigma,W)$ be an expression, and let $\vec{P},\vec{Q}\in \T(\Sigma)^W$ be $W$-tuples of
  processes for some $W \subseteq \Var$, such that $\vec{P}(x)\sqsubseteq_B\vec{Q}(x)$ for each $x\in W$.
  Then $E[\vec{P}] \sqsubseteq_B E[\vec{Q}]$.
\end{theoremA}
In \cite{vG17b} this result was actually obtained for the more general ntyft/ntyxt-format with recursion,
obtained by lifting the restriction ruling out lookahead. The special case that $E$ is a
recursion-free expression stems from \cite{BolG96} (under a side condition on the TSS called
\emph{well-foundedness}, which was lifted in \cite{FG96}). The special subcase in which $\TS$ has no
negative premises goes back to \cite{GrV92}.

\textsc{Milner} \cite{Mi90ccs} proposed an important tool in proving bisimilarity between two processes: the
notion of a bisimulation up to bisimilarity. Here I extend it to the bisimulation preorder on
3-valued TSSs.

\begin{definitionA}{Bisimulation upto}{upto}
Let $\TS = (\Sigma,\R)$ be a TSS\@.
A \emph{bisimulation $\B$ up to bisimilarity} is a binary relation on $\T(\Sigma)$ such that, for
$P,Q\in\T(\Sigma)$ and $a\in A$,
\begin{itemize}
\item if $P\B Q$ and $P \goto{a} P'$, then there is a $Q'$ with $Q \goto{a} Q'$ and
  $P' \sqsubseteq_B \,\B\, \sqsubseteq_B Q'$,
\item if $P\B Q$ and $Q \hoto{a} Q'$, then there is a $P'$ with $P \hoto{a} P'$ and
  $P' \sqsubseteq_B \,\B\, \sqsubseteq_B Q'$.
\end{itemize}
Here
${\sqsubseteq_B\,\B\,\sqsubseteq_B} := \{(R,S) \mid \exists R',S'.~R\sqsubseteq_B R' \B S'\sqsubseteq_B S\}$.
\end{definitionA}

\begin{proposition}{upto}
If $P \B Q$ for some bisimulation $\B$ up to bisimilarity, then $P \sqsubseteq_B Q$.
\end{proposition}
\begin{proof}
Using the reflexivity of $\sqsubseteq_B$ it suffices to show that $\sqsubseteq_B \,\B\, \sqsubseteq_B$
is a bisimulation.
Using transitivity of $\sqsubseteq_B$ this is straightforward.
\end{proof}

\section{The recursive specification principle}\label{sec:RSP}

Let $\Sigma$ be a signature.
Any preorder ${\sqsubseteq} \subseteq \T(\Sigma)\times\T(\Sigma)$ on the closed terms over $\Sigma$
extends to a preorder ${\sqsubseteq} \subseteq \IT(\Sigma)\times\IT(\Sigma)$ on open terms
by defining $E \sqsubseteq F$ iff $E[\rho] \sqsubseteq F[\rho]$ for each closed substitution
$\rho$. %:\Var \rightarrow\T(\Sigma)$. 
It extends to substitutions $\rho,\nu:\Var\rightharpoonup \IT(\Sigma)$
by $\rho\sim \nu$ iff $\textrm{dom}(\rho) = \textrm{dom}(\nu)$ and
$\rho(x) \sim \nu(x)$ for each $x\in \textrm{dom}(\rho)$.
As $W$-tuples $\vec{E},\vec{F}\in \IT(\Sigma)^W$ for $W \subseteq \Var$ are substitutions,
this gives $\vec{E}\sqsubseteq\vec{F}$ iff $\vec{E}(x) \sqsubseteq \vec{F}(x)$ for each $x \in W$.
A recursive specification $\RS: V_\RS \rightarrow  \IT(\Sigma)$ is a $V_\RS$-tuple;
if $\vec{E}\in \IT(\Sigma)^{V_\RS}$ is another $V_\RS$-tuple then with
$\RS[\vec{E}]$ I denote the $V_\RS$-tuple with elements $\RS_X[\vec{E}]$ for each $X \in V_\RS$,
that is, the body $\RS_X$ of the recursive equation of $X$ in which $\vec{E}(Y)$ is substituted for
each variable $Y \in V_\RS$ occurring freely in $\RS_X$.

\begin{definitionA}{Solution}{solution}
For $\RS$ a recursive specification and $\equiv$ a semantic equivalence,
a $V_\RS$-tuple of expressions $\vec{E} \in \IT(\Sigma)^{V_\RS}$ is a \emph{solution} of $\RS$ if
$\vec{E} \equiv \RS[\vec{E}]$.
\end{definitionA}
When generalising this definition with a preorder $\sqsubseteq$ in the r\^ole of $\equiv$, one
obtains \emph{pre-} and \emph{postsolutions}.
\begin{definition}{presolution}
For $\RS$ a recursive specification and $\sqsubseteq$ a semantic equivalence,
a $V_\RS$-tuple of expressions $\vec{E} \in \IT(\Sigma)^{V_\RS}$ is a \emph{presolution}
(resp.\ \emph{postsolution}) of $\RS$ if $\vec{E} \sqsubseteq \RS[\vec{E}]$ (resp.\
$\RS[\vec{E}] \sqsubseteq \vec{E}$).
\end{definition}
The following theorem establishes RSP for the bisimulation preorder.

\begin{theorem}{RSP}
  Let $\TS=(\Sigma,\R)$ be a TSS in the ready simulation format with recursion,
  and $\RS$ be a $\TS$-guarded recursive specification.
If $\vec{E} \sqsubseteq_B \RS[\vec{E}]$ and $\RS[\vec{F}] \sqsubseteq_B \vec{F}$ with
$\vec E, \vec F\mathbin\in\IT(\Sigma)^{V_\RS}$, then $\vec{E} \sqsubseteq_B \vec{F}$.
\end{theorem}

\begin{proof}
It suffices to prove \thm{RSP} under the assumptions that $\vec E, \vec F\mathbin\in\T(\Sigma)^{V_\RS}$ and
only the variables from $V_\RS$ occur free in the expressions $\RS_X$ for $X \in V_\RS$.
For in the general case I have to establish that $\vec{E}[\sigma] \sqsubseteq_B \vec{F}[\sigma]$
for an arbitrary closed substitution $\sigma:\Var\rightarrow\T(\Sigma)$.
Let $\sigma^\dagger:\Var{\setminus}V_\RS\rightarrow\T(\Sigma)$ be given by $\sigma^\dagger(X)=\sigma(X)$ for
all $X\in \Var{\setminus}V_\RS$. Then $\vec{E} \sqsubseteq_B \RS[\vec{E}]$ and
$\RS[\vec{F}] \sqsubseteq_B \vec{F}$ imply
$\vec{E}[\sigma] \sqsubseteq_B \RS[\vec{E}][\sigma] = \RS[\sigma^\dagger][\vec{E}[\sigma]]$
and $\RS[\sigma^\dagger][\vec{F}[\sigma]] = \RS[\vec{F}][\sigma] \sqsubseteq_B \vec{F}[\sigma]$.
Hence, I merely have to prove the theorem with $\vec{E}[\sigma]$, $\vec{F}[\sigma]$and
$\RS[\sigma^\dagger]$ in place of $\vec{E}$, $\vec{F}$ and $\RS$.

It also suffices to prove \thm{RSP} under the assumption that $\RS$ is a manifestly guarded
recursive specification. Namely, for a general guarded recursive specification $\RS$, let $\RS'$ be
the manifestly guarded recursive specification into which $\RS$ can be converted.
Then $\vec{E} \sqsubseteq_B \RS[\vec{E}]$ implies
$\vec{E} \sqsubseteq_B \RS'[\vec{E}]$ by \thm{congruence} and transitivity, and likewise
$\RS[\vec{F}] \sqsubseteq_B \vec{F}$ implies $\RS'[\vec{F}] \sqsubseteq_B \vec{F}$.

So let $\RS$ be manifestly guarded with free variable from $V_\RS$ only,
and let $\vec{P},\vec{Q}\in \T^{V_\RS}(\Sigma)$ be a pre- and postsolution of $\RS$, that
is, $\vec{P} \sqsubseteq_B \RS[\vec{P}]$ and $\RS[\vec{Q}] \sqsubseteq_B \vec{Q}$.
I will show that \[{\B} := \{G [\RS[\vec{P}]], G[\RS[\vec{Q}]] \mid
G\mathbin\in\IT(\Sigma) \mbox{~is recursion-free and has variables from $V_\RS$ only}\}\] is a
bisimulation up to bisimilarity. Once I have that, taking $G := X\in V_\RS$ yields
$\RS_X[\vec{P}] \sqsubseteq_B \RS_X[\vec{Q}]$ by \pr{upto},
and thus $\vec{P}(X) \sqsubseteq_B \RS_X[\vec{P}] \sqsubseteq_B \RS_X[\vec{Q}] \sqsubseteq_B \vec{Q}(X)$
for all $X\in V_\RS$. So $\vec{P} \sqsubseteq_B \vec{Q}$.

\begin{itemize}
\item Let $R \B T$ and $R \goesto{\al} R'$ with $\al\mathbin\in A$.
I have to find a $T'$ with $T \goesto{\al} T'$ and $R' \sqsubseteq_B\, \B \,\sqsubseteq_B T'\!$.
Let $R \mathbin= G[\RS[\vec{P}]]$ and $T\mathbin=G[\RS[\vec{Q}]]$, where $G\mathbin\in\IT(\Sigma)$
is recursion-free and has variables from $V_\RS$ only.
Note that $G[\RS[\vec{P}]]$ can also be written as $G[\RS][\vec{P}]$.
Since the expressions $\RS_X$ for $X\in V_\RS$ have free variables from $V_\RS$ only, so does $G[\RS]$.
Moreover, since $\RS$ is manifestly guarded, the expression $G[\RS]$ must be guarded.
By \lem{guarded}, $R'$ must have the form $G'[\vec{P}]$, where $G'\in\IT(\Sigma,V_\RS)$.
Moreover, $T = G [\RS[\vec{Q}]] = G[\RS][\vec{Q}] \goesto\al G'[\vec{Q}] =: T'$. 
Furthermore, by \thm{congruence}, $G'[\vec{P}] \sqsubseteq_B G'[\RS[\vec{P}]]$ and
$G'[\RS[\vec{Q}]] \sqsubseteq_B G'[\vec{Q}]$. By definition $G'[\RS[\vec{P}]] \B G'[\RS[\vec{Q}]]$.
Thus, $R'\mathbin= G'[\vec{P}]\sqsubseteq_B\,\B\,\sqsubseteq_B G'[\vec{Q}] \mathbin= T'\!$.
\item Let $R \B T$ and $T \hoto{\al} T'$ with $\al\mathbin\in A$.
I have to find an $R'$ with $R \hoto{\al} R'$ and $R' \sqsubseteq_B\, \B \,\sqsubseteq_B T'\!$.
Let $R \mathbin= G[\RS[\vec{P}]]$ and $T\mathbin=G[\RS[\vec{Q}]]$, where $G\mathbin\in\IT(\Sigma)$
is recursion-free and has variables from $V_\RS$ only.
Note that $G[\RS[\vec{Q}]]$ can also be written as $G[\RS][\vec{Q}]$.
Since the expressions $\RS_X$ for $X\in V_\RS$ have free variables from $V_\RS$ only, so does $G[\RS]$.
Moreover, since $\RS$ is manifestly guarded, the expression $G[\RS]$ must be guarded.
By \lem{guarded}, $T'$ must have the form $G'[\vec{Q}]$, where $G'\in\IT(\Sigma,V_\RS)$.
Moreover, $R = G [\RS[\vec{P}]] = G[\RS][\vec{P}] \hoto\al G'[\vec{P}] =: R'$. 
Furthermore, by \thm{congruence}, $G'[\vec{P}] \sqsubseteq_B G'[\RS[\vec{P}]]$ and
$G'[\RS[\vec{Q}]] \sqsubseteq_B G'[\vec{Q}]$. By definition $G'[\RS[\vec{P}]] \B G'[\RS[\vec{Q}]]$.
Thus, $R'\mathbin= G'[\vec{P}]\sqsubseteq_B\,\B\,\sqsubseteq_B G'[\vec{Q}] \mathbin= T'\!$.\qed
\end{itemize}
\end{proof}

\section{The recursive definition principle}\label{sec:RDP}
The recursive definition principle, saying that recursive specifications have a solution, trivially
holds for bisimulation equivalence, and thereby also for all coarser equivalences. Namely the
recursive call $\rec{X|\RS}$ itself is a solution of the recursive specification $\RS$.
\begin{equation}
  \rec{X|\RS} \equiv_B \rec{\RS_X|\RS} \quad \mbox{for all~} X \in V_\RS
  \label{RDP}
\end{equation}
This follows immediately from the operational rules for recursion.

Write $\rec{V_\RS|\RS}$ for the $V_\RS$-tuple with elements $\rec{X|\RS}$ for $X \in V_\RS$.
Connecting this with the notation of \df{substitutions}, and seeing $V_\RS$-tuples as substitutions,
this yields $\rec{E|\RS} = E[\rec{V_\RS|\RS}]$ for any $E\in \IT(\Sigma)$.
Moreover, $\RS[\rec{V_\RS|\RS}]$ is the $V_\RS$-tuple with elements $\rec{\RS_X|\RS}$ for $X \in V_\RS$.

Now (\ref{RDP}), or RDP, for an arbitrary semantic equivalence $\equiv$, can be restated as
$\rec{V_\RS|\RS} \equiv \RS[\rec{V_\RS|\RS}]$. When $\equiv$ is the \hyperlink{kernel}{kernel}
of a preorder $\sqsubseteq$, this implies $\RS[\rec{V_\RS|\RS}] \sqsubseteq \rec{V_\RS|\RS}$ and
$\rec{V_\RS|\RS} \sqsubseteq  \RS[\rec{V_\RS|\RS}]$.
Using this, RSP (\ref{RSP}) can be restated as
\begin{equation}\label{RDSP}
 \frac{\vec{P} \equiv \RS[\vec{P}]}{\vec{P} \equiv \rec{V_\RS|\RS}}
\qquad
 \frac{\vec{P} \sqsubseteq \RS[\vec{P}]}{\vec{P} \sqsubseteq \rec{V_\RS|\RS}}
\qquad
 \frac{\RS[\vec{Q}] \sqsubseteq \vec{Q}}{\rec{V_\RS|\RS} \sqsubseteq \vec{Q}}\;.
\end{equation}

\section{The simulation and ready simulation preorders}

The ready simulation preorder $\sqsubseteq_{\it RS}$ \cite{BIM95,vG01} has been generalised to
3-valued transition systems in \cite{BFG04}. Its definition is just like \df{refinement}, but
without the targets $Q'$ and $P'$ in the second clause. Here $Q \hoto{a}$ means that there exists a
$Q'$ such that $Q \hoto{a} Q'$.  In \cite{BFG04} it is also shown to be a congruence
(cf.\ \thm{congruence}) on TSSs in the ready simulation format, albeit without recursion. Using
this, the proof of \thm{RSP} generalises smoothly to $\sqsubseteq_{\it RS}$. There merely is a
straightforward simplification in the second clause, dealing with transitions $\hoto{}$. I conclude
that RSP holds for the ready simulation preorder for TSSs in the ready simulation format.

The simulation preorder $\sqsubseteq_S$ \cite{vG01} is defined just like the
\hyperlink{bisimilarity}{bisimulation preorder},
but without its second clause. It is a congruence only for TSSs in the positive
ready simulation format, that is, ruling out negative premises \cite{BFG04}.
As a counterexample, in the TSS of \ex{TSS}, but extended with the familiar choice operator $+$ of
CCS \cite{Mi90ccs}, we have $a.0 \sqsubseteq_S a.0+b.0$, yet
$\en_{\{b\}}(a.0) \not\sqsubseteq_S \en_{\{b\}}(a.0+b.0)$,
since only $\en_{\{b\}}(a.0)$ can do first a $b$ and then an $a$-transition.
Again a straightforward simplification of the above proof shows that
RSP holds for the simulation preorder for TSSs in the positive ready simulation format.

\begin{example}{RSP for simulation preorder}
Let $\RS := \{X = a.X + b.0\}$. Take $P := \rec{Y|\{Y=a.Y\}}$ and $Q := \rec{Z|\{Z=a.Z\ + a.b.0\}}$.

\expandafter\ifx\csname graph\endcsname\relax
   \csname newbox\expandafter\endcsname\csname graph\endcsname
\fi
\ifx\graphtemp\undefined
  \csname newdimen\endcsname\graphtemp
\fi
\expandafter\setbox\csname graph\endcsname
 =\vtop{\vskip 0pt\hbox{%
\pdfliteral{
q [] 0 d 1 J 1 j
0.576 w
0.576 w
10.224 -48.96 m
10.224 -50.948225 8.612225 -52.56 6.624 -52.56 c
4.635775 -52.56 3.024 -50.948225 3.024 -48.96 c
3.024 -46.971775 4.635775 -45.36 6.624 -45.36 c
8.612225 -45.36 10.224 -46.971775 10.224 -48.96 c
S
Q
}%
    \graphtemp=.5ex
    \advance\graphtemp by 0.500in
    \rlap{\kern 0.092in\lower\graphtemp\hbox to 0pt{\hss $P$\hss}}%
\pdfliteral{
q [] 0 d 1 J 1 j
0.576 w
0.072 w
q 0 g
10.152 -58.896 m
9.144 -51.48 l
13.464 -57.528 l
10.152 -58.896 l
B Q
0.576 w
4.104 -51.48 m
1.8 -57.42 l
0.23328 -61.4592 -0.504 -65.664 -0.504 -70.56 c
-0.504 -75.456 1.78848 -77.76 6.66 -77.76 c
11.53152 -77.76 13.824 -75.456 13.824 -70.56 c
13.824 -65.664 13.12128 -61.5744 11.628 -57.78 c
9.432 -52.2 l
S
Q
}%
    \graphtemp=.5ex
    \advance\graphtemp by 1.180in
    \rlap{\kern 0.092in\lower\graphtemp\hbox to 0pt{\hss $a$\hss}}%
\pdfliteral{
q [] 0 d 1 J 1 j
0.576 w
82.224 -12.96 m
82.224 -14.948225 80.612225 -16.56 78.624 -16.56 c
76.635775 -16.56 75.024 -14.948225 75.024 -12.96 c
75.024 -10.971775 76.635775 -9.36 78.624 -9.36 c
80.612225 -9.36 82.224 -10.971775 82.224 -12.96 c
S
Q
}%
    \graphtemp=.5ex
    \advance\graphtemp by 0.000in
    \rlap{\kern 1.092in\lower\graphtemp\hbox to 0pt{\hss $a.P+b.0$\hss}}%
\pdfliteral{
q [] 0 d 1 J 1 j
0.576 w
111.024 -27.36 m
111.024 -29.348225 109.412225 -30.96 107.424 -30.96 c
105.435775 -30.96 103.824 -29.348225 103.824 -27.36 c
103.824 -25.371775 105.435775 -23.76 107.424 -23.76 c
109.412225 -23.76 111.024 -25.371775 111.024 -27.36 c
S
Q
}%
    \graphtemp=.5ex
    \advance\graphtemp by 0.190in
    \rlap{\kern 1.312in\lower\graphtemp\hbox to 0pt{\hss $b$\hss}}%
\pdfliteral{
q [] 0 d 1 J 1 j
0.576 w
0.072 w
q 0 g
98.568 -20.952 m
104.184 -25.776 l
96.984 -24.12 l
98.568 -20.952 l
B Q
0.576 w
81.864 -14.544 m
97.776 -22.536 l
S
82.224 -48.96 m
82.224 -50.948225 80.612225 -52.56 78.624 -52.56 c
76.635775 -52.56 75.024 -50.948225 75.024 -48.96 c
75.024 -46.971775 76.635775 -45.36 78.624 -45.36 c
80.612225 -45.36 82.224 -46.971775 82.224 -48.96 c
S
0.072 w
q 0 g
80.424 -38.16 m
78.624 -45.36 l
76.824 -38.16 l
80.424 -38.16 l
B Q
0.576 w
78.624 -16.56 m
78.624 -38.16 l
S
Q
}%
    \graphtemp=.5ex
    \advance\graphtemp by 0.430in
    \rlap{\kern 1.092in\lower\graphtemp\hbox to 0pt{\hss ~~~~$a$\hss}}%
\pdfliteral{
q [] 0 d 1 J 1 j
0.576 w
0.072 w
q 0 g
82.152 -58.896 m
81.144 -51.48 l
85.464 -57.528 l
82.152 -58.896 l
B Q
0.576 w
76.104 -51.48 m
73.764 -57.42 l
72.1728 -61.4592 71.424 -65.664 71.424 -70.56 c
71.424 -75.456 73.728 -77.76 78.624 -77.76 c
83.52 -77.76 85.824 -75.456 85.824 -70.56 c
85.824 -65.664 85.12128 -61.5744 83.628 -57.78 c
81.432 -52.2 l
S
Q
}%
    \graphtemp=.5ex
    \advance\graphtemp by 1.180in
    \rlap{\kern 1.092in\lower\graphtemp\hbox to 0pt{\hss $a$\hss}}%
\pdfliteral{
q [] 0 d 1 J 1 j
0.576 w
154.224 -48.96 m
154.224 -50.948225 152.612225 -52.56 150.624 -52.56 c
148.635775 -52.56 147.024 -50.948225 147.024 -48.96 c
147.024 -46.971775 148.635775 -45.36 150.624 -45.36 c
152.612225 -45.36 154.224 -46.971775 154.224 -48.96 c
S
Q
}%
    \graphtemp=.5ex
    \advance\graphtemp by 0.500in
    \rlap{\kern 2.092in\lower\graphtemp\hbox to 0pt{\hss $Q$\hss}}%
\pdfliteral{
q [] 0 d 1 J 1 j
0.576 w
0.072 w
q 0 g
154.152 -58.896 m
153.144 -51.48 l
157.464 -57.528 l
154.152 -58.896 l
B Q
0.576 w
148.104 -51.48 m
145.764 -57.42 l
144.1728 -61.4592 143.424 -65.664 143.424 -70.56 c
143.424 -75.456 145.728 -77.76 150.624 -77.76 c
155.52 -77.76 157.824 -75.456 157.824 -70.56 c
157.824 -65.664 157.12128 -61.5744 155.628 -57.78 c
153.432 -52.2 l
S
Q
}%
    \graphtemp=.5ex
    \advance\graphtemp by 1.180in
    \rlap{\kern 2.092in\lower\graphtemp\hbox to 0pt{\hss $a$\hss}}%
\pdfliteral{
q [] 0 d 1 J 1 j
0.576 w
183.024 -63.36 m
183.024 -65.348225 181.412225 -66.96 179.424 -66.96 c
177.435775 -66.96 175.824 -65.348225 175.824 -63.36 c
175.824 -61.371775 177.435775 -59.76 179.424 -59.76 c
181.412225 -59.76 183.024 -61.371775 183.024 -63.36 c
S
Q
}%
    \graphtemp=.5ex
    \advance\graphtemp by 0.690in
    \rlap{\kern 2.312in\lower\graphtemp\hbox to 0pt{\hss $b$\hss}}%
\pdfliteral{
q [] 0 d 1 J 1 j
0.576 w
0.072 w
q 0 g
170.568 -56.952 m
176.184 -61.776 l
168.984 -60.12 l
170.568 -56.952 l
B Q
0.576 w
153.864 -50.544 m
169.776 -58.536 l
S
211.824 -77.76 m
211.824 -79.748225 210.212225 -81.36 208.224 -81.36 c
206.235775 -81.36 204.624 -79.748225 204.624 -77.76 c
204.624 -75.771775 206.235775 -74.16 208.224 -74.16 c
210.212225 -74.16 211.824 -75.771775 211.824 -77.76 c
S
Q
}%
    \graphtemp=.5ex
    \advance\graphtemp by 0.890in
    \rlap{\kern 2.712in\lower\graphtemp\hbox to 0pt{\hss $c$\hss}}%
\pdfliteral{
q [] 0 d 1 J 1 j
0.576 w
0.072 w
q 0 g
199.368 -71.352 m
204.984 -76.176 l
197.784 -74.52 l
199.368 -71.352 l
B Q
0.576 w
182.664 -64.944 m
198.576 -72.936 l
S
262.224 -12.96 m
262.224 -14.948225 260.612225 -16.56 258.624 -16.56 c
256.635775 -16.56 255.024 -14.948225 255.024 -12.96 c
255.024 -10.971775 256.635775 -9.36 258.624 -9.36 c
260.612225 -9.36 262.224 -10.971775 262.224 -12.96 c
S
Q
}%
    \graphtemp=.5ex
    \advance\graphtemp by 0.000in
    \rlap{\kern 3.592in\lower\graphtemp\hbox to 0pt{\hss $a.Q+b.0$\hss}}%
\pdfliteral{
q [] 0 d 1 J 1 j
0.576 w
291.024 -27.36 m
291.024 -29.348225 289.412225 -30.96 287.424 -30.96 c
285.435775 -30.96 283.824 -29.348225 283.824 -27.36 c
283.824 -25.371775 285.435775 -23.76 287.424 -23.76 c
289.412225 -23.76 291.024 -25.371775 291.024 -27.36 c
S
Q
}%
    \graphtemp=.5ex
    \advance\graphtemp by 0.190in
    \rlap{\kern 3.812in\lower\graphtemp\hbox to 0pt{\hss $b$\hss}}%
\pdfliteral{
q [] 0 d 1 J 1 j
0.576 w
0.072 w
q 0 g
278.568 -20.952 m
284.184 -25.776 l
276.984 -24.12 l
278.568 -20.952 l
B Q
0.576 w
261.864 -14.544 m
277.776 -22.536 l
S
262.224 -48.96 m
262.224 -50.948225 260.612225 -52.56 258.624 -52.56 c
256.635775 -52.56 255.024 -50.948225 255.024 -48.96 c
255.024 -46.971775 256.635775 -45.36 258.624 -45.36 c
260.612225 -45.36 262.224 -46.971775 262.224 -48.96 c
S
0.072 w
q 0 g
260.424 -38.16 m
258.624 -45.36 l
256.824 -38.16 l
260.424 -38.16 l
B Q
0.576 w
258.624 -16.56 m
258.624 -38.16 l
S
Q
}%
    \graphtemp=.5ex
    \advance\graphtemp by 0.430in
    \rlap{\kern 3.592in\lower\graphtemp\hbox to 0pt{\hss ~~~~$a$\hss}}%
\pdfliteral{
q [] 0 d 1 J 1 j
0.576 w
0.072 w
q 0 g
262.152 -58.896 m
261.144 -51.48 l
265.464 -57.528 l
262.152 -58.896 l
B Q
0.576 w
256.104 -51.48 m
253.764 -57.42 l
252.1728 -61.4592 251.424 -65.664 251.424 -70.56 c
251.424 -75.456 253.728 -77.76 258.624 -77.76 c
263.52 -77.76 265.824 -75.456 265.824 -70.56 c
265.824 -65.664 265.12128 -61.5744 263.628 -57.78 c
261.432 -52.2 l
S
Q
}%
    \graphtemp=.5ex
    \advance\graphtemp by 1.180in
    \rlap{\kern 3.592in\lower\graphtemp\hbox to 0pt{\hss $a$\hss}}%
\pdfliteral{
q [] 0 d 1 J 1 j
0.576 w
291.024 -63.36 m
291.024 -65.348225 289.412225 -66.96 287.424 -66.96 c
285.435775 -66.96 283.824 -65.348225 283.824 -63.36 c
283.824 -61.371775 285.435775 -59.76 287.424 -59.76 c
289.412225 -59.76 291.024 -61.371775 291.024 -63.36 c
S
Q
}%
    \graphtemp=.5ex
    \advance\graphtemp by 0.690in
    \rlap{\kern 3.812in\lower\graphtemp\hbox to 0pt{\hss $b$\hss}}%
\pdfliteral{
q [] 0 d 1 J 1 j
0.576 w
0.072 w
q 0 g
278.568 -56.952 m
284.184 -61.776 l
276.984 -60.12 l
278.568 -56.952 l
B Q
0.576 w
261.864 -50.544 m
277.776 -58.536 l
S
319.824 -77.76 m
319.824 -79.748225 318.212225 -81.36 316.224 -81.36 c
314.235775 -81.36 312.624 -79.748225 312.624 -77.76 c
312.624 -75.771775 314.235775 -74.16 316.224 -74.16 c
318.212225 -74.16 319.824 -75.771775 319.824 -77.76 c
S
Q
}%
    \graphtemp=.5ex
    \advance\graphtemp by 0.890in
    \rlap{\kern 4.212in\lower\graphtemp\hbox to 0pt{\hss $c$\hss}}%
\pdfliteral{
q [] 0 d 1 J 1 j
0.576 w
0.072 w
q 0 g
307.368 -71.352 m
312.984 -76.176 l
305.784 -74.52 l
307.368 -71.352 l
B Q
0.576 w
290.664 -64.944 m
306.576 -72.936 l
S
Q
}%
    \hbox{\vrule depth1.180in width0pt height 0pt}%
    \kern 4.442in
  }%
}%

\centerline{\box\graph}
\vspace{2ex}

\noindent
Now $P \sqsubseteq_S  \RS[P] = a.P + b.0$ and $\RS[Q] = a.Q + b.0 \sqsubseteq_S  Q$.
Hence, by RSP, $P \sqsubseteq_S Q$.
\end{example}

\noindent
The literature is divided about the orientation of the simulation preorder, i.e., whether
$a.0 \sqsubseteq_S a.0+b.0$ or $a.0+b.0 \sqsubseteq_S a.0$, and likewise for all semantic preorders.
Intuitively, the argument for $a.0 \sqsubseteq_S a.0+b.0$ is that there exists a simulation
\emph{from} $a.0$ \emph{to} $a.0+b.0$, in the sense that behaviour of the former process can be
mimicked by the latter. Related to that, the process $a.0$ has \emph{less} behaviours or runs than $a.0+b.0$.

Below I will follow the argument of \textsc{Tony Hoare} \cite{Ho85} that $P \sqsubseteq Q$ should
mean that $Q$ is a (non-strict) \emph{improvement} of $P$, in the sense that all good properties of $P$
should hold for $Q$ as well. As a specification $S$ tells us the minimal set of requirements that an
implementation $I$ should satisfy, this pleads for the orientation $S \sqsubseteq I$.

Usually this argument is used to advocate the direction $a.0+b.0 \sqsubseteq_S a.0$.  A good
property could be a \emph{satisfy} property, saying that nothing bad will ever happen, or a
\emph{liveness} property, saying that something good will happen eventually \cite{Lam77}. An example
of a safety property could be that the action $b$ will never occur, and an example of a liveness
property could be that each run contains the action $a$ or $b$. Both safety and liveness properties
are \emph{linear time properties}: they are defined to hold for certain runs of a system, and then
are said to hold for a system iff they hold for all its runs. Now a system with less runs will
automatically satisfy more linear time properties, and thus be a more suitable implementation. This
pleads for the direction $a.0+b.0 \sqsubseteq_S a.0$.

Nevertheless, theoretically one could have quite different intentions with specifications, modelled
as 2- or 3-valued labelled transition systems. It could be that the specification stipulates which
transitions should be present at least, allowing room for adding more transitions. This fits with
allowing \emph{possibility properties} \cite{Lam98} in the above argument about good properties
preserved by semantic preorders. Such properties may say that a process has the possibility to behave
in a certain way.

In \cite{vG93}, for each semantic preorder $\sqsubseteq$ surveyed there I propose two variants,
called $\sqsubseteq^{\rm may}$ and $\sqsubseteq^{\rm must}\!$. Although these preorders differ from
each other in non-trivial ways when dealing with divergence and underspecification, when restricted
to processes $P$ and $Q$ without divergence and underspecification, we simply have
$P \sqsubseteq^{\rm may} Q$ iff $Q \sqsubseteq^{\rm must} P$. For the simulation preorder we have
$a.0 \sqsubseteq^{\rm may}_S a.0+b.0$ and $a.0+b.0 \sqsubseteq^{\rm must}_S a.0$, and more in
general $P \sqsubseteq^{\rm may} Q$ says that possible behaviour of $P$ must also be possible for $Q$,
whereas $P \sqsubseteq^{\rm must} Q$ says that required behaviour of $P$ must also be required of $Q$.
The annotations $^{\rm may}$ and $^{\rm must}$ were inspired by the may- and must-testing preorders
of \textsc{De Nicola \& Hennessy} \cite{DH84}, which also are oriented this way.

When applying preorders $\sqsubseteq^{\rm may}$ and $\sqsubseteq^{\rm must}$ to modal transition
systems, in which transitions $\goto{a}$ are required, and transitions $\hoto{a}$ are allowed,
it makes most sense to define $\sqsubseteq^{\rm may}$ in terms of the necessary transitions $\goto{a}$,
and $\sqsubseteq^{\rm must}$ in terms of the possible transitions $\hoto{a}$.
Thus $\sqsubseteq^{\rm may}_S$ is defined like $\sqsubseteq_B$ in \df{refinement}, but without the
second clause, whereas $\sqsubseteq^{\rm must}_S$ is defined like $\sqsubseteq_B$, but without the
first clause. In the same spirit, $\sqsubseteq^{\rm may}_{\it RS}$ is as $\sqsubseteq_B$ but without
the targets in the second clause, so that this second clause can be restated as
\begin{itemize}
\item if $P\B Q$ and $P \gonotto{a}$, then $Q \gonotto{a}$.
\end{itemize}
Similarly, $\sqsubseteq^{\rm must}_{\it RS}$ is as $\sqsubseteq_B$ but without the targets in the
first clause. Hence the preorders $\sqsubseteq_{\it RS}$ and $\sqsubseteq_S$ from \cite{vG01,BFG04} that
were featured at the beginning of this section are the may-versions of the ready simulation
and simulation preorders. The must-versions of these preorders, which may be preferable when
concentrating on safety and liveness properties, satisfy RSP by entirely symmetric arguments.

\begin{corollary}{RSP}
  RSP holds for $\sqsubseteq^{\rm must}_{\it RS}$ for any TSS in the ready simulation format.\\
  RSP also holds for $\sqsubseteq^{\rm must}_S$ for any TSS in the positive ready simulation format.
\end{corollary}

\section{Full congruence results for guarded recursion}\label{sec:full}

A semantic equivalence $\equiv$ is called a \emph{full congruence} for recursion \cite{vG17b}
if \[\RS \equiv \RS' \Rightarrow \rec{V_\RS|\RS} \equiv \rec{V_\RS|\RS'}\] for all recursive
specifications $\RS$ and $\RS'$ with $V_\RS = V_\RS'$. When expanding the application of $\equiv$ to
$V_\RS$-tuples, defined at the beginning of \Sec{RSP}, and the tuple $\rec{V_\RS|\RS}$ of \Sec{RDP},
this translates to
\begin{center}
  if $\RS_Y \equiv \RS'_Y$ for all $Y \in V_\RS$ then $\rec{X|\RS} \equiv \rec{X|\RS'}$ for all $X
  \in \RS$,
\end{center}
where $\RS_Y \equiv \RS'_Y$ means that $\RS_Y[\sigma] \equiv \RS'_Y[\sigma]$ for all closed
substitutions $\sigma$. In \cite{vG17b} this was merely required for recursive
specifications $\RS,\RS':V_\RS \rightarrow \IT(\Sigma,V_\RS)$, that is, without free variables
(cf.\ \ex{recursion}), so that $\rec{X|\RS}$ and $\rec{X|\RS'}$ are closed terms. However, when we
have that, the general form follows trivially, just as in the first paragraph of the proof of \thm{RSP}.
In the same way, a preorder $\sqsubseteq$ is a \emph{full precongruence} for recursion \cite{vG17b}
if\vspace{-1ex}
 \[\RS \sqsubseteq \RS' \Rightarrow \rec{V_\RS|\RS} \sqsubseteq \rec{V_\RS|\RS'}\] for all recursive
specifications $\RS$ and $\RS'$ with $V_\RS = V_\RS'$. A preorder $\sqsubseteq$ is a
\emph{precongruence} for an $n$-ary operator $f \in \Sigma$ if
\[\vec{P} \sqsubseteq \vec{Q} \Rightarrow f(\vec{P}) \sqsubseteq f(\vec{Q})\]
where $\vec{P}$ and $\vec{Q}$ denote the tuples $(P_i)_{i=1}^n$ and $(Q_i)_{i=1}^n$.
Clearly, if $\sqsubseteq$ is a precongruence for all operators $f$ in a language given by a
signature $\Sigma$, as well as a full precongruence for recursion, then it surely is a
\emph{lean precongruence} \cite{vG17b} (cf.\ \thm{congruence}), meaning that
\[\vec{P} \sqsubseteq \vec{Q} \Rightarrow E[\vec{P}]\sqsubseteq E[\vec{Q}]\]
for all expressions $E\in \IT(\Sigma)$. Moreover, being a lean precongruence implies that a preorder
is a precongruence for all operators. However, as illustrated in \cite{vG17b}, being a lean
precongruence does not imply that a preorder is a full precongruence for recursion.
In \cite{vG17b} it was shown that (strong) bisimilarity is a lean precongruence for all TSSs in the
ready simulation format with recursion---in fact, this was shown for the even more general
ntyft/ntyxt format with recursion. However, a full precongruence result could be obtained only for
the \emph{positive} ntyft/ntyxt format with recursion, avoiding all negative premises. For the ready
simulation format with recursion, allowing negative premises, and even for the more restrictive GSOS
format (\df{GSOS}), it remained an open question whether bisimilarity is a full
precongruence. Here I partly answer that question: for TSSs in the ready simulation format with
recursion, bisimilarity is a full precongruence for \emph{guarded} recursion. This is a
straightforward consequence of RSP\@.

\begin{theorem}{full}
  Let $\TS=(\Sigma,\R)$ be a TSS in the ready simulation format with recursion,
  and $\RS,\RS'$ be two recursive specifications with $V_\RS=V_{\RS'}$ and one of them $\TS$-guarded.
  If $\RS \sqsubseteq_B \RS'$ then $\rec{V_\RS|\RS} \sqsubseteq_B \rec{V_\RS|\RS'}$.
\end{theorem}
\begin{proof}
  I consider first the case that $\RS$ and then the one that $\RS'$ is guarded.

  By RDP, $\rec{V_\RS|\RS'} \equiv_B \RS'[\rec{V_\RS|\RS'}]$, and thus
  $\RS'[\rec{V_\RS|\RS'}] \sqsubseteq_B \rec{V_\RS|\RS'}$.
  Since $\RS[\sigma] \sqsubseteq_B \RS'[\sigma]$ for all substitutions $\sigma$,
  in particular $\RS[\rec{V_\RS|\RS'}] \sqsubseteq_B \RS'[\rec{V_\RS|\RS'}]$.
  By transitivity, $\RS\rec{V_\RS|\RS'} \sqsubseteq_B \rec{V_\RS|\RS'}$, that is, $\rec{V_\RS|\RS'}$
  is a postsolution of $\RS$.
  Given that $\RS$ is guarded, this yields $\rec{V_\RS|\RS} \sqsubseteq_B \rec{V_\RS|\RS'}$ by RSP (\ref{RDSP}).

  By RDP, $\rec{V_\RS|\RS} \equiv_B \RS[\rec{V_\RS|\RS}]$, and thus
  $\rec{V_\RS|\RS} \sqsubseteq_B \RS[\rec{V_\RS|\RS}]$.
  Since $\RS[\sigma] \sqsubseteq_B \RS'[\sigma]$ for all substitutions $\sigma$,
  in particular $\RS[\rec{V_\RS|\RS}] \sqsubseteq_B \RS'[\rec{V_\RS|\RS}]$.
  By transitivity, $\rec{V_\RS|\RS} \sqsubseteq_B \RS'\rec{V_\RS|\RS}$, that is, $\rec{V_\RS|\RS}$
  is a presolution of $\RS'$.
  Given that $\RS'$ is guarded, this yields $\rec{V_\RS|\RS} \sqsubseteq_B \rec{V_\RS|\RS'}$ by RSP (\ref{RDSP}).
\end{proof}
Similarly, bisimulation equivalence, $\equiv_B$, is a full congruence for guarded recursion.
This can be shown in exactly the same way, or follows directly from \thm{full}, using that
$\equiv_B$ is the \hyperlink{kernel}{kernel} of $\sqsubseteq_B$.

Also in the same way it follows that $\sqsubseteq^{\rm may}_{\it RS}$ and
$\sqsubseteq^{\rm must}_{\it RS}$ are full precongruences for guarded recursion on TSSs
in the ready simulation format with recursion, and that $\sqsubseteq^{\rm may}_{\it S}$ and
$\sqsubseteq^{\rm must}_{\it S}$ are full precongruences for guarded recursion on TSSs
in the positive ready simulation format with recursion.

\section{GSOS languages}

\begin{definition}{GSOS}
A \emph{GSOS rule} is a ntyft/ntyxt rule without lookahead, such that the left-hand sides of all
premises are variables that also occur in the source of the rule, and its target contains only
variables that also occur in its source or premises. A TSS is in \emph{GSOS format} if has the
recursion rules from \df{ntyft}, and all of its other rules are recursion-free GSOS rules.
A \emph{GSOS language} is language whose syntax is given by a signature $\Sigma$ and its semantics by
TSS $\TS$ in GSOS format, but only allowing recursive calls $\rec{X|\RS}$ for $\TS$-guarded
recursive specifications $\RS$.
\end{definition}
GSOS languages, or \emph{GSOS rule systems}, were introduced by \textsc{Bloom, Istrail \& Meyer} in
\cite{BIM88}. Here ``GSOS'' stands for \emph{structural operational semantics with guarded recursion}.
In \cite{BIM88} the action prefixing operators $a.\_$ of CCS and \ex{TSS} are required to be present
in each GSOS language, so that guarded recursion could be defined as in CCS, using only such
operators as guards. As here I have proposed a general definition of guardedness for GSOS languages,
the requirement that the operators $a.\_$ must be present can be dropped.
In \cite{BIM95}, which could be seen as the journal version of \cite{BIM88}, recursion has been
dropped altogether, and there the ``G'' of ``GSOS'' `might as well stand for ``Grand.'''

\begin{theorem}{2-valued GSOS}
The well-founded semantics of a GSOS language is always 2-valued.
\end{theorem}

\begin{proof}
  Let $\TS=(\Sigma,\R)$ be the TSS of a GSOS language. If $\rec{CT^+,PT^+}$ is its well-founded
  semantics, I have to show that $CT^+ = PT^+$.
  This amounts to $PT^+ \subseteq CT^+$, as the other direction holds always.

Let $\rightsquigarrow$ be the smallest binary relation on $\T(\Sigma)$ such that
(i) $f(P_1,\dots,P_k) \rightsquigarrow P_i$ for each $k$-ary operator $f\in \Sigma$ and each
$i\in\{1,\dots,k\}$ such that the $i^{\rm th}$ argument of $f$ is unguarded, and
(ii) $\rec{X|\RS} \rightsquigarrow \rec{\RS_X|\RS}$.

\textit{Claim:} $\rightsquigarrow$ has no forward infinite chain
$P_0 \rightsquigarrow P_1 \rightsquigarrow \dots$.

\textit{Proof of claim:} Let $P_0 \rightsquigarrow P_1 \rightsquigarrow \dots$ be a forward infinite
chain.  First consider the possibility that $P_{h_0}$ has the form $\rec{X|\RS}$ for some $h_0\geq 0$,
and all $P_m$ with $m>h_0$ have a subterm of the form $\rec{Y|\RS}$. In that case there must
exist indices $h_i$ for $i \in \IN$ such that $h_0 < h_1 < h_2 < \cdots$ and, for all $i \in \IN$,
each $P_{h_i}$ has the form $\rec{X_i|\RS}$ for some $X_i \in V_\RS$, and $X_{i+1}$ occurs unguarded
in the expression $P_{h_i+1} = \rec{\RS_{X_i} | \RS}$.\footnote{Here I use that a term
$\rec{\RS_{X_i} | \RS}$ cannot contain, as a subterm, a recursive call $\rec{Z|\RS'}$ with
$\RS'\neq \RS$, such that $\rec{Z|\RS'}$ has a subterm of the form $\rec{Y|\RS}$. This is
because terms and recursive specifications are constructed inductively, and if $\RS$ is a recursive
specification with a subterm $\rec{Z|\RS'}$ then $\RS'$ cannot have a subterm $\rec{Y|\RS}$.} It
follows that the recursive specification $\RS$ is not manifestly guarded, and neither can it be
converted into a manifestly guarded recursive specification by repeated substitution of expressions
$\RS_Y$ for variables $Y\in V_\RS$ occurring in expressions $\RS_Z$ for $Z\in V_\RS$.
This contradicts the assumption that all recursive specifications $\RS$ in the language must be guarded.
It follows that each $P_h$ in the chain is followed by an expression $P_m$ with $m>h$ that is a
strict subterm of $P_h$. But this contradicts the inductive process of building terms.

\textit{Application of the claim:}
The relation $\rightsquigarrow$ is finitely branching, meaning that for any $P \in \T(\Sigma)$ there
are only finitely many $P^\dagger$ with $P \rightsquigarrow P^\dagger$.
For any process $P\in\T(\Sigma)$, let $e(P)$ be the length of the longest forward chain
$P \rightsquigarrow P_1 \rightsquigarrow \dots \rightsquigarrow P_{e(P)}$.
I show with induction on $e(P)$ that
\[(P \goto{a} P')\in PT^+ ~~\Rightarrow~~ (P \goto{a} P')\in CT^+\]
which suffices to conclude the proof.
Here I make a case distinction on the shape of $P$.

Let $P=f(P_1,\dots,P_{\ar{f}})$. Using that
\hyperlink{note}{$\beta \in PT^+$ iff $\TS\vdash \frac{PT^-}{\beta}$},
the transition $(P \goto{a} P')\in PT^+$ must be derivable from $PT^-$
through a substitution instance of a proof rule $r=\frac{H}{f(x_1,\dots,x_{\ar{f}}) \goto\al F}$ from $\R$.
Since $r$ is a GSOS rule, the premises of that substitution instance must have the form
$P_i \goto{a} R$ and/or $P_i \gonotto{a}$ with $i\in\{1,\dots,\ar{f}\}$.
Moreover, in case $i$ is be a guarded argument of $f$, the rule $r$ cannot have premises $x_i \goto{a} y$
or $x_i \gonotto{a}$, using that the occurrence of the variable $x$ in the term $x$ is unguarded.
Thus, if there is a premise $P_i \goto{a} R$ or $P_i \gonotto{a}$ then $i$ must be an unguarded
argument of $f$, and hence $P \rightsquigarrow P_i$. In that case $e(P_i) < e(P)$, and I may apply the
induction hypothesis.

If the substitution instance of $r$ has a premise $P_i \goto{a} R$, then $(P_i \goto{a} R) \in PT^+$,
and by induction $(P_i \goto{a} R) \in CT^+$. If it has a premise $P_i \gonotto{a}$, then
$(P_i \gonotto{a}) \in PT^-$. I show that $(P_i \gonotto{a}) \in CT^-$ by contradiction.
Namely, if $(P_i \gonotto{a}) \notin CT^-$, then $(P_i \goto{a} Q) \in PT^+$
for some $Q \in \T(\Sigma)$, by \obs{well-founded}. In that case $(P_i \goto{a} Q) \in CT^+$ by
induction, which implies $(P_i \gonotto{a}) \notin PT^-$, again by \obs{well-founded}, so that a
contradiction is reached.

Applying the rule $r$ again, I obtain $(P \goto{a} P')\in CT^+$.

In case $P=\rec{X|\RS}$, the transition $(P \goto{a} P')\in PT^+$ must be derivable from $PT^-$
using a proof rule for recursion with premise $\rec{\RS_X|\RS} \goto{a} P'$.
Since $e(\rec{\RS_X|\RS}) < e(\rec{X|\RS})$ I may apply the induction hypothesis.
So from $(\rec{\RS_X|\RS} \goto{a} P') \in PR^+$ I conclude $(\rec{\RS_X|\RS} \goto{a} P') \in CR^+$
by induction, and this yields $(P \goto{a} P')\in CT^+$.
\end{proof}

\noindent
In \cite{BIM88} it was observed that processes in GSOS languages $(\Sigma,\R)$ are finitely branching,
in the sense that for each $P \in \T(\Sigma)$ there are only finitely many pairs
$(a,P^\dagger)\in A\times\T(\Sigma)$ such that $P \goto{a} P^\dagger$. Here I show that this
observation remains valid now that I have generalised the definition of guardedness.
In~\cite{BIM88,BIM95} GSOS languages are required to feature only finitely many GSOS rules, and each of
these rules has only finitely many premises. In the following I need a weaker assumption, namely
that each operator $f \in \Sigma$ occurs in the source of only finitely many abstract GSOS rules,
and each of these rules has only finitely many positive premises. Here I define an \emph{abstract
GSOS rule} as a (possibly infinite) family of GSOS rules that differ only in their transition
labels, and with the property that for each choice of transition labels in the premises there are
only finitely many choices for the transition label of the conclusion.\pagebreak[3]
GSOS languages that satisfy this requirement I here call \emph{finitary}.

\begin{proposition}{finitely branching}
Each process in a finitary GSOS language is finitely branching.
\end{proposition}

\begin{proof}
  Let $\TS=(\Sigma,\R)$ be the TSS of a finitary GSOS language.  Again I apply the metric $e(P)$ defined
  in the previous proof.  I show with induction on $e(P)$ that there are only finitely many proofs $\pi$
  of transitions with source $P$ that may use negative literals from $CT^-$ as assumptions. This is
  sufficient, since
  \hyperlink{note}{\mbox{$(P \goto{a} P')\in CT^+$ iff $\TS \vdash \frac{CT^-}{P \goto{a} P'}$}}.
  If $\pi$ derives a transition $Q \goesto{a} R$, then I call $Q$ the \emph{source} of $\pi$.

In case $P=f(P_1,\dots,P_{\ar{f}})$, a proof $\pi$ with source $P$ is completely determined by the
concluding abstract rule from $\R$, the label of the concluding transition, and the
subderivations of $\pi$ with source $P_i$ for some of the $i\in\{1,\dots,\ar{f}\}$.  Moreover, as
already argued in the previous proof, such subderivations only exists when $i$ is an unguarded
argument of $f$, so that $e(P_i)<e(P)$.  The choice of the concluding abstract rule $r$ depends on $f$,
and for each $f$ there are only finitely many choices. Moreover, using that $e(P_i)<e(P)$, by induction
there are only finitely many subderivations of $\pi$ with source $P_i$. Finally, for each collection
of premises of an instance of $r$ (determined by these subderivations) there are only finitely many
possibilities for the label of the concluding transition. This shows there are only finitely
many choices for $\pi$.

In case $P=\rec{X|\RS}$, the last step in $\pi$ must be application of the rules for recursion, so
$\pi$ is completely determined by a subderivation $\pi'$ of a transition with source $\rec{\RS_X|\RS}$.
By induction there are only finitely many choices for $\pi'$, and hence also for $\pi$.
\end{proof}

\section{A complete axiomatisation of bisimilarity for GSOS languages}\label{sec:complete}

Let BCCSP \cite{vG01} be the common kernel of the process algebras CCS \cite{Mi90ccs} and CSP
\cite{BHR84}, consisting of the operators $0$, $a.\_$ for $a \in A$ and $\_ + \_$.
It is well-known \cite{Mi90ccs,BW90,Fok00} that the collection~{\it Ax} of the four axioms A1--4
constitutes a sound and ground-complete axiomatisation of bisimilarity over BCCSP\@.

\begin{minipage}{2.5in}
  \[\begin{array}{c@{\qquad}r@{~=~}l}
    \rm A1 & x+y   & y+x \\
    \rm A2 & (x+y)+z & x+(y+z) \\
    \rm A3 & x+x   & x \\
    \rm A4 & x+0   & x \\
  \end{array}\]
  \vspace{1pt}
\end{minipage}
\hfill
$\displaystyle\frac{x \goto{a} y}{\pi_{n+1} \goto{a} \pi_n(y)}$
\hfill
\begin{minipage}{2.5in}
  \[\begin{array}{r@{~=~}l}
    \pi_0(x) & 0 \\
    \pi_{n+1}(0) & 0 \\
    \pi_{n+1}(a.x) & a.\pi_n(x) \\
    \pi_{n+1}(x+y) & \pi_{n+1}(x) +  \pi_{n+1}(y)\\
  \end{array}\]
  \vspace{1pt}
\end{minipage}

\noindent
This means that $P \equiv_B Q \Leftrightarrow \textit{Ax} \vdash P=Q$ for all BCCSP processes $P$
and $Q$. Here provability is w.r.t.\ equational logic.
To reason in an algebraic way about BCCSP processes extended with guarded recursion, the
\emph{Approximation Induction Principle} (AIP) \cite{BBK87a,BW90,Fok00} can be used. It employs
\emph{projection operators} $\pi_n$ for all $n\in\IN$ that cut off a process after $n$ transitions.
There transition rules are displayed above.
Complete axioms for these operators, when added to BCCSP with recursion, are displayed above as well.
Now AIP says that equivalence of processes is determined by the equivalence of their finite
projections. It can be formulated as
\[\frac{\forall n.~~ \pi_n(x) = \pi_n(y)}{x=y}\;.\]
As shown, e.g., in \cite{Fok00}, \textit{Ax} together with RDP and AIP, and of course the projection
axioms, constitutes a sound and ground-complete axiomatisation of BCCSP with guarded recursion.
Here, provability is w.r.t.\ infinitary conditional equational logic \cite{GV93}.

\textsc{Aceto, Bloom \& Vaandrager} \cite{ABV94} extended this sound and ground-complete
axiomatisation to finite GSOS languages $(\Sigma,\R)$ that contain BCCSP as a sublanguage.
Their key contribution is to provide for each such language a set of axioms that allow to convert
each process into a head-normal form.

\begin{definition}{hnf}
A \emph{head formal form} is a process of the form $\Sigma_{i \in I} a_i.P_i$ with $a_i\in A$ and
$P_i\in \T(\Sigma)$ for all $i$ from a finite index set $I$.
\end{definition}
In generating these head-normalising axioms some auxiliary operators need to be added to language.
In \cite{ABV94} only recursion-free GSOS languages are studied with finitely many GSOS rules, in
which each rule has only finite many premises. However, their method of generating axioms that
allow to convert each process into a head-normal form generalises straightforwardly to finitary
GSOS languages, allowing guarded recursion, now also using RDP among the axioms. As a consequence,
\textit{Ax} together with RDP, AIP, the projection axioms and the head-normalising axioms from
\cite{ABV94} constitutes a sound and ground-complete axiomatisation of any finitary GSOS language
that contains BCCSP as a sublanguage.

I here contribute that in the above ground-complete axiomatisation result one can replace AIP by RSP\@.
In \cite[Section~4.4]{Fok00} it is shown that the basic axioms that are ground-complete for recursion-free
ACP processes, together with RDP and RSP, fail to constitute a ground-complete axiomatisation of
bisimilarity for ACP with guarded recursion. However, this is because in \cite{Fok00} only recursive
specifications with finitely many recursion equations are allowed. When allowing arbitrary many
recursive equations, the counterexample of \cite{Fok00} goes away. In fact, \cite[Theorem 4.4.1]{Fok00}
shows that the basic axioms that are ground-complete for recursion-free ACP processes,
together with RDP and RSP, do form a ground-complete axiomatisation for ACP processes with finite,
linear, guarded recursive specifications. Here \emph{linear} says in essence that within these
recursive specifications only BCCSP operators are allowed (except that in the setting of ACP one
employs a variant that distinguishes deadlock from successful termination). As observed in
\cite{GM20}, nothing in the proof of \cite[Theorem 4.4.1]{Fok00} depends on the finiteness of
recursive specifications, so the result holds without that disclaimer. This does turn RSP into an
infinitary conditional equation, just like AIP, so that also here infinitary conditional equational
logic needs to be employed.

\textsc{Van Glabbeek \& Middelburg} \cite{GM20} also show that using RSP, any guarded recursive
specification can be converted into a linear, guarded recursive specification. This shows that the
basic axioms that are ground-complete for recursion-free ACP processes, together with RDP and RSP,
do form a sound and ground-complete axiomatisation of ACP with guarded recursion. Moreover, nothing
in the reasoning of \cite{GM20} is specific to the language ACP\@. The only relevant fact about ACP
that is incorporated in the proof is \cite[Proposition~1]{GM20}, saying that each ACP process with
guarded recursion can be brought into head normal form.
Thus, combining the insights of \cite{ABV94,Fok00,GM20}, it follows that \textit{Ax} together with
RDP, RSP and the head-normalising axioms from \cite{ABV94} constitutes a sound and ground-complete
axiomatisation of any finitary GSOS language that contains BCCSP as a sublanguage.

\section{Future work}

It would be great to extend this work to branching bisimilarity and/or other weak equivalences.

\bibliographystyle{eptcs}
\bibliography{../../Biblio/abbreviations,../../Biblio/dbase,../../Biblio/new,glabbeek}
\end{document}